\documentclass[runningheads]{llncs}
\usepackage[utf8]{inputenc}
\usepackage[T1]{fontenc}    

\usepackage{pgf,tikz,pgfplots}
\pgfplotsset{compat=1.14}
\usetikzlibrary{arrows}

\usepackage{amsmath}
\usepackage{amsfonts}
\usepackage{amssymb}
\usepackage{thmtools}
\usepackage{thm-restate}

\def\commentsversion{0}  

\ifnum\commentsversion=0 
\newcommand{\comments}[1]{}
\fi

\def\newversion{1}  

\ifnum\newversion=0 
\newcommand{\old}[1]{#1}
\newcommand{\new}[1]{}
\fi

\ifnum\newversion=1 
\newcommand{\old}[1]{}
\newcommand{\new}[1]{\textcolor{blue}{this is a proposed version. let me know if its ok to adopt or not.~~}#1}
\fi

\usepackage{listings}

\ifnum\commentsversion=1 
\newcommand{\comments}[1]{#1}
\fi
        {\hspace*{\fill}$\Box$\par}

\let\emptyset\varnothing

\tikzset{every picture/.style={line width=0.75pt}} 

\title{Maximum Weight Disjoint Paths in Outerplanar Graphs via Single-Tree Cut Approximators}
\titlerunning{Maximum Weight EDP in Outerplanar Graphs}

\author{Guyslain Naves\inst{1} \and Bruce Shepherd\inst{2} \and Henry Xia\inst{2}}
\authorrunning{G. Naves et al.}

\institute{Aix-Marseille University\thanks{work partially supported by ANR project DISTANCIA (ANR-17-CE40-0015).}, LIS, CNRS UMR 7020, \email{guyslain.naves@univ-amu.fr} 
\and University of British Columbia \email{fbrucesh@cs.ubc.ca}, \email{h.xia@alumni.ubc.ca}
}

\date{\today}

\begin{document}

\maketitle

\begin{abstract}
    Since 1997 there has been a steady stream of advances for the maximum  disjoint paths problem. Achieving tractable results has usually required focusing on relaxations such as: (i) to allow some bounded edge congestion in solutions, (ii) to only consider the unit weight (cardinality) setting, (iii) to only require fractional routability of the selected demands (the all-or-nothing flow setting). For the general form (no congestion, general weights, integral routing) of edge-disjoint paths ({\sc edp}) even the case of unit capacity trees which are stars generalizes the maximum matching problem for which Edmonds provided an exact algorithm.  For general capacitated trees, Garg, Vazirani, Yannakakis showed the problem is APX-Hard and Chekuri, Mydlarz, Shepherd provided a $4$-approximation. This is essentially the only setting where a  constant approximation is known for the general form of \textsc{edp}.   We extend their result by giving  a constant-factor approximation algorithm for general-form \textsc{edp} in  outerplanar graphs. A key component for the algorithm is to find a {\em single-tree} $O(1)$ cut approximator for outerplanar graphs. Previously $O(1)$ cut approximators were only known  via distributions on trees and these were based implicitly on the  results of Gupta, Newman, Rabinovich and Sinclair for distance tree embeddings combined with results of  Anderson and Feige.
\end{abstract}

\newpage

\section{Introduction}

The past two decades have seen numerous advances to the approximability of the maximum disjoint paths problem ({\sc edp}) since the seminal paper \cite{GargVY97}. 
An instance of \textsc{edp} consists of  a (directed or undirected) ``supply'' graph $G=(V,E)$ and a collection of $k$ {\em requests} (aka demands). Each request consists of a pair of nodes $s_i,t_i \in V$. These
are sometimes viewed as a {\em demand graph} $H=(V(G).\{s_it_i: i \in [k]\})$. A subset $S$ of the requests is called {\em routable} if there exist edge-disjoint paths $\{P_i: i \in S\}$  such that $P_i$ has endpoints $s_i,t_i$ for each $i$. We may also be given a profit $w_i$ associated with each request and the goal
is to find a  routable subset $S$ which maximizes $w(S)=\sum_{i \in S} w_i$. The {\em cardinality version} is  where
we have unit weights $w_i \equiv 1$.

For directed graphs it is known \cite{guruswami2003near} that there is no $O(n^{0.5-\epsilon})$
approximation, for any $\epsilon >0$ under the assumption $P \neq NP$. Subsequently, research  shifted to undirected graphs
and two relaxed models. First, in the {\em  all-or-nothing flow model} ({\sc anf})   the notion of routability is relaxed. A subset $S$ is called routable if there is a feasible (fractional) multiflow which satisfies each request in $S$.    In \cite{Chekuri04a} a polylogarithmic approximation is given  for {\sc anf}.  Second, in the {\em congestion} model \cite{KleinbergT98} one is allowed to increase the capacity of each edge in $G$ by some constant factor.
Two streams of results ensued. For general graphs,  a polylogarithmic approximation is ultimately provided \cite{chuzhoy2012polylogarithimic,ChuzhoyL12,chekuri2013poly}  with edge congestion $2$.    For planar graphs, a constant factor approximation is given \cite{seguin2020maximum,CKS-planar-constant}  with edge congestion $2$. There is also an $f(g)$-factor approximation for bounded genus $g$ graphs with congestion 3.

As far as we know, the only congestion $1$ results known for either maximum {\sc anf} or {\sc edp} are as follows; all of these apply only to the cardinality version.
In \cite{kawarabayashi2018all},  a constant factor approximation is given for {\sc anf} in planar graphs and  
for treewidth $k$ graphs there is an $f(k)$-approximation for {\sc edp} \cite{chekuri2013maximum}.
More recent results include a constant-factor approximation in the {\em fully planar} case where $G+H$ is planar \cite{huang2020approximation,garg2020integer}. 
In the weighted regime, there is
a factor $4$ approximation  for  
{\sc edp} in capacitated trees \cite{chekuri2007multicommodity}. We remark that this problem for unit capacity ``stars'' already generalizes the maximum weight matching problem in general graphs. Moreover, inapproximability bounds for {\sc edp} in planar graphs are almost polynomial \cite{chuzhoy2017new}. This lends interest to how far one  can push  beyond trees.   Our main contribution to the theory of maximum throughput flows is the following result which is the first  generalization  of the (weighted) {\sc edp} result for trees 
\cite{chekuri2007multicommodity},
modulo a larger implicit constant  of $224$.

\begin{restatable}{theorem}{outerplanarWEDPapprox}
\label{thm:edp}
 There is an $224$ approximation algorithm for 
the maximum weight {\sc anf} and {\sc edp} problems for capacitated
outerplanar graphs. 
\end{restatable}

It is natural to  try to prove this is by reducing the problem in outerplanar graphs to trees and then use \cite{chekuri2007multicommodity}.
A promising approach is to use  results of
\cite{gupta2004cuts} -- an $O(1)$ distance tree embedding for outerplanar graphs --  and a {\em transfer theorem} \cite{andersen2009interchanging,Racke08} which proves a general  equivalence between distance and capacity embeddings.
Combined, these results  imply that there is a probabilistic  embedding into trees which approximates cut capacity in outerplanar graphs with constant congestion.   
One could then try to mimic the success of  using low-distortion (distance) tree  embeddings to approximate  minimum cost  network design problems. There is an issue with this approach however. Suppose we have a distribution on trees $T_i$ which approximates cut capacity in expectation.  We then apply a known {\sc edp} algorithm which outputs a subset of requests $S_i$ which are routable in each $T_i$.    While the tree embedding guarantees the convex combination of $S_i$'s satisfies the cut condition in $G$, it may be that no single $S_i$ obeys the cut condition, even approximately. This is a problem  even for {\sc anf}.  In fact, this seems to be a problem even when  the trees are  either dominating  or dominated by $G$.
We resolve this   by computing a {\bf single} tree which approximates the cuts in $G$ -- see Theorem~\ref{thm:tree}. Our algorithmic proof  is heavily inspired by  work of Gupta \cite{gupta2001steiner} which gives a method for eliminating Steiner nodes in probabilistic (distance) tree embeddings for general graphs.  

It turns out that having a single-tree is not enough for us and we need additional technical properties to apply the algorithm from \cite{chekuri2007multicommodity}.  First,   our single tree $T$ should have  integer capacities and be  non-expansive, i.e., $\hat{u}(\delta_T(S)) \leq u(\delta_G(S))$ (where $\hat{u}/u$ are the edge capacities in $T/G$ and $\delta$ is used to denote the edges in the cut induced by $S$).
To see why it is useful that $T$ is an under-estimator of $G$'s cut capacity,     consider the classical grid example of \cite{GargVY97}. They give an instance with a set of $\sqrt{n}$ requests which satisfy the cut condition in $2 \cdot G$, but for which  one can only route  a single request in the capacity of $G$.

If our tree is an under-estimator, then  we can ultimately obtain a ``large'' weight subset of requests  satisfying the cut condition in $G$ itself. However, even this is not generally  sufficient for (integral) routability. For a multiflow instance $G/H$ one normally also requires that $G+H$ is Eulerian,
even for easy instances such as when $G$  is a $4$-cycle. The final ingredient we use  is  that our single tree $T$ is actually a {\bf subtree} of $G$ which
 allows us to invoke  the following result -- see Section~\ref{sec:required}. 
\begin{restatable}{theorem}{OProute}
\label{thm:OP}
Let G be an outerplanar graph with integer edge capacities $u(e)$. Let $H$ denote a
demand graph such that $G + H = (V (G),E(G) \cup E(H))$ is outerplanar. If $G,H$ satisfies the cut
condition, then $H$ is routable in G.
\end{restatable} 

\noindent 
The key point here is that  we can avoid the usual parity condition  needed, such as in \cite{Okamura81,seymour1981matroids,frank1985edge}.
We  are not presently aware of the above result's existence in the  literature.

\subsection{A Single-Subtree Cut Sparsifier and Related Results}

Our main cut approximation theorem is the following which may be of independent interest.

\begin{restatable}{theorem}{integerTree}
\label{thm:tree}
For any connected outerplanar graph $G=(V,E)$ with integer edge capacities $u(e) > 0$, there is a subtree $T$ of $G$ with integer edge weights $\hat{u}(e) \geq 0$ such that
\[
\frac{1}{14} u(\delta_G(X)) \leq \hat{u}(\delta_{T}(X)) \leq u(\delta_G(X)) \mbox{ for each proper subset $X \subseteq V$} 
\]
\end{restatable}

We discuss some connections of this result to prior work on sparsifiers and metric embeddings.
 Celebrated work of R\"acke \cite{racke02} shows the existence of a single capacitated tree $T$ (not a subtree) which behaves as a  flow sparsifier for a given graph $G$.  In particular, 
 routability of demands on $T$ implies fractional routability in $G$ with edge congestion $polylog(n)$; this bound was further improved to $O(\log^2n \log\log n)$ \cite{harrelson2003polynomial}.  Such single-tree results were also instrumental in an  application to maximum throughput flows:  a polylogarithmic approximation for the maximum all-or-nothing flow problem in general graphs \cite{chekuri2013all}. Even more directly to Theorem~\ref{thm:tree}  is work on cut sparsifiers; in \cite{racke2014improved} it is shown that there is a single tree  (again, not subtree) which approximates cut capacity in a general tree $G$ within a factor of  $O(\log^{1.5} \log\log n)$.  As far as we know, our result is the only global-constant factor single-tree cut approximator for a family of graphs.

 R\"acke improved the bound for flow sparsification to an optimal congestion of $O(\log n)$ \cite{Racke08}. Rather than a single tree, this work requires a convex combination  of (general) trees to simulate the capacity in $G$. His  work also revealed a beautiful equivalence between the existence of good (low-congestion) distributions over trees for capacities, and 
 the existence of good  (low-distortion) distributions over trees for distances \cite{andersen2009interchanging}. 
 This {\em transfer theorem} states very roughly that for a graph $G$ the following are equivalent for a given $\rho \geq 1$. (1) For any edge lengths $\ell(e)>0$, there is a (distance) embedding of $G$ into a distribution of trees which has stretch at most $\rho$. (2) For any edge capacities $u(e)>0$, there is a (capacity) embedding of $G$ into a distribution of trees which has congestion at most $\rho$.  This work has been applied in other related contexts such as flow sparsifiers for proper subsets of terminals \cite{englert2014vertex}.
 
 The transfer theorem uses a very general setting where there are a collection of valid {\em maps}. A  map $M$ sends an  edge of $G$  to an abstract ``path'' $M(e) \subseteq E(G)$.  The maps may be refined for the application of interest. In the so-called {\em spanning tree setting}, each $M$ is associated with a subtree $T_M$ of $G$ (the setting most relevant to Theorem~\ref{thm:tree}). $M(e)$ is then the unique path  which joins the endpoints of $e$ in $T_M$.  For an edge $e$, its {\em stretch} under $M$ is  $(\sum_{e' \in <(e)} \ell(e'))/\ell(e)$.
 In the context of distance tree embeddings this model has been studied in \cite{alon1995graph,AbrahamBN08,elkin2008lower}.
 In capacity settings, the {\em congestion} of an edge under $M$ is $(\sum_{e': e \in M(e)} c(e'))/c(e)$.  One can view this as simulating the capacity of $G$ using the tree's edges with bounded congestion.  The following result shows that we cannot guarantee a single subtree with $O(1)$ congestion even for outerplanar graphs;   this example was found independently  by Anastasios Sidiropoulos \cite{tasos}.

\begin{theorem}
\label{thm:lowerbound}
There is an infinite family $\mathcal{O}$ of outerplanar graphs  
such that for every  $G \in \mathcal{O}$ and every spanning tree $T$ of $G$:
\[
\max_{X} \frac{u(\delta_G(X))}{u(\delta_T(X))} = \Omega(\log|V(G)|),
\]
where the max is taken over fundamental cuts of $T$.
\end{theorem}

This suggests that the single-subtree result Theorem~\ref{thm:tree} is  a bit  lucky
and critically requires the use of tree capacities different from $u$. 
Of course  a single tree is sometimes
unnecessarily restrictive. For instance, outerplanar graphs also have an $O(1)$-congestion embedding using a distribution of subtrees by the transfer theorem (although we are not aware of one  explicitly given in the literature). This follows implicitly due to existence of an $O(1)$-stretch embedding into subtrees  \cite{gupta2004cuts}. 

Finally we remark that despite the connections between distance and capacity tree embeddings,  Theorem~\ref{thm:tree} stands in contrast to the situation for distance embeddings.  Every embedding of the $n$ point cycle into a (single) subtree suffers  distortion $\Omega(n)$, and indeed this also holds for embedding into an arbitrary (using Steiner nodes etc.) tree \cite{rabinovich1998lower}.

\section{Single spanning tree cut approximator in Outerplanar Graphs}

In this section we first show the existence of a single-tree 
which is an $O(1)$ cut approximator for an outerplanar graph $G$.
Subsequently we show that there is such a tree with two additional properties. First, its capacity on every cut is at most the capacity in $G$, and second, all of its weights are integral.  These additional properties (integrality and conservativeness) are needed in our application to {\sc edp}. The formal statement we prove is as follows.

\integerTree*

In Section~\ref{sec:flowdist}, we show how to view capacity approximators in $G$ as (constrained) distance tree approximators in the planar dual graph. From then on, we look for distance approximators in the dual which correspond to trees in $G$. In Section~\ref{sec:non-conservative} we prove there exists a  single-subtree cut approximator. In Appendix~\ref{sec:extend} we show how to make this conservative while maintaining integrality of the capacities.  
 In Section~\ref{sec:lb} we show that we cannot achieve Theorem~\ref{thm:tree} in the exact weight model.

\subsection{Converting flow-sparsifiers in outerplanar graphs to distance-sparsifiers in trees}
\label{sec:flowdist}

Let $G = (V, E)$ be an outerplanar graph with capacities $u:E\to\mathbb{R}^+$.
Without loss of generality, we can assume that $G$ is 2-node connected,
so the boundary of the outer face of $G$ is a cycle that
contains each node exactly once. Let $G^*$ be the dual of $G$; we assign weights
to the dual edges in $G^*$ equal to the capacities on the corresponding edges in $G$.
Let $G_z$ be the graph obtained by adding an apex node $z$ to $G$ which is connected
to each node of $G$, that is $V(G_z)=V\cup\{z\}$ and
$E(G_z)=E\cup\{(z,v):v\in V\}$. We may embed $z$ into the outer face of $G$, so $G_z$
is planar. Let $G_z^*$ denote the planar dual of $G_z$.

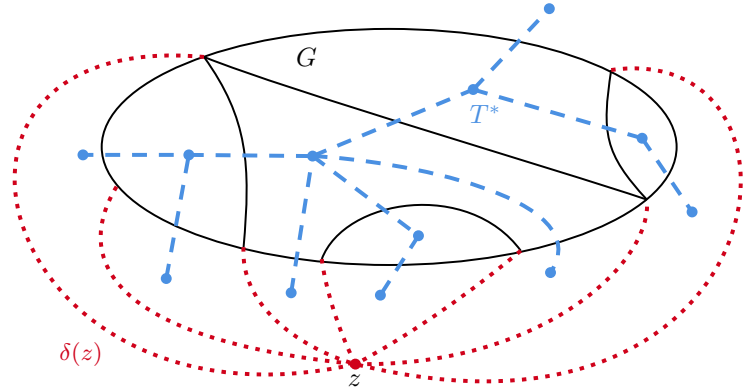
\begin{figure}
\centering
\begin{tikzpicture}[x=0.75pt,y=0.75pt,yscale=-1,xscale=1,scale=0.6]

\draw   (101,132.3) .. controls (101,77.46) and (209.18,33) .. (342.63,33) .. controls (476.08,33) and (584.26,77.46) .. (584.26,132.3) .. controls (584.26,187.14) and (476.08,231.6) .. (342.63,231.6) .. controls (209.18,231.6) and (101,187.14) .. (101,132.3) -- cycle ;
\draw    (187,56.6) .. controls (224.26,108.6) and (226.26,147.6) .. (220.26,216.6) ;
\draw    (187,56.6) .. controls (259.26,85.6) and (482.26,150.6) .. (559.26,176.6) ;
\draw    (286,226.6) .. controls (305.26,175.6) and (409.26,159.6) .. (453.26,219.6) ;
\draw    (529.26,67.6) .. controls (521.26,120.6) and (523.26,137.6) .. (559.26,176.6) ;
\draw [color={rgb, 255:red, 208; green, 2; blue, 27 }  ,draw opacity=1 ][line width=1.5]  [dash pattern={on 1.69pt off 2.76pt}]  (187,56.6) .. controls (-83.07,31.81) and (11.26,382.98) .. (314.26,314.98) ;
\draw [color={rgb, 255:red, 208; green, 2; blue, 27 }  ,draw opacity=1 ][line width=1.5]  [dash pattern={on 1.69pt off 2.76pt}]  (220.26,216.6) .. controls (216.93,273.29) and (281.26,306.07) .. (314.26,314.98) ;
\draw [color={rgb, 255:red, 208; green, 2; blue, 27 }  ,draw opacity=1 ][line width=1.5]  [dash pattern={on 1.69pt off 2.76pt}]  (286,226.6) .. controls (287.33,238.28) and (290.75,252.71) .. (295.23,267.2) .. controls (300.61,284.59) and (307.54,302.06) .. (314.26,314.98) ;
\draw [color={rgb, 255:red, 208; green, 2; blue, 27 }  ,draw opacity=1 ][line width=1.5]  [dash pattern={on 1.69pt off 2.76pt}]  (453.26,219.6) .. controls (413.93,252.29) and (362.93,289.29) .. (314.26,314.98) ;
\draw [color={rgb, 255:red, 208; green, 2; blue, 27 }  ,draw opacity=1 ][line width=1.5]  [dash pattern={on 1.69pt off 2.76pt}]  (314.26,314.98) .. controls (469.5,317.67) and (564.5,230.67) .. (559.26,176.6) ;
\draw [color={rgb, 255:red, 208; green, 2; blue, 27 }  ,draw opacity=1 ][line width=1.5]  [dash pattern={on 1.69pt off 2.76pt}]  (529.26,67.6) .. controls (714.5,40.81) and (676.5,405.67) .. (314.26,314.98) ;
\draw  [color={rgb, 255:red, 208; green, 2; blue, 27 }  ,draw opacity=1 ][fill={rgb, 255:red, 208; green, 2; blue, 27 }  ,fill opacity=1 ] (310.51,314.98) .. controls (310.51,312.91) and (312.19,311.23) .. (314.26,311.23) .. controls (316.34,311.23) and (318.01,312.91) .. (318.01,314.98) .. controls (318.01,317.05) and (316.34,318.73) .. (314.26,318.73) .. controls (312.19,318.73) and (310.51,317.05) .. (310.51,314.98) -- cycle ;
\draw  [color={rgb, 255:red, 74; green, 144; blue, 226 }  ,draw opacity=1 ][fill={rgb, 255:red, 74; green, 144; blue, 226 }  ,fill opacity=1 ] (274.51,139.98) .. controls (274.51,137.91) and (276.19,136.23) .. (278.26,136.23) .. controls (280.34,136.23) and (282.01,137.91) .. (282.01,139.98) .. controls (282.01,142.05) and (280.34,143.73) .. (278.26,143.73) .. controls (276.19,143.73) and (274.51,142.05) .. (274.51,139.98) -- cycle ;
\draw  [color={rgb, 255:red, 74; green, 144; blue, 226 }  ,draw opacity=1 ][fill={rgb, 255:red, 74; green, 144; blue, 226 }  ,fill opacity=1 ] (170.51,138.98) .. controls (170.51,136.91) and (172.19,135.23) .. (174.26,135.23) .. controls (176.34,135.23) and (178.01,136.91) .. (178.01,138.98) .. controls (178.01,141.05) and (176.34,142.73) .. (174.26,142.73) .. controls (172.19,142.73) and (170.51,141.05) .. (170.51,138.98) -- cycle ;
\draw  [color={rgb, 255:red, 74; green, 144; blue, 226 }  ,draw opacity=1 ][fill={rgb, 255:red, 74; green, 144; blue, 226 }  ,fill opacity=1 ] (409.51,83.98) .. controls (409.51,81.91) and (411.19,80.23) .. (413.26,80.23) .. controls (415.34,80.23) and (417.01,81.91) .. (417.01,83.98) .. controls (417.01,86.05) and (415.34,87.73) .. (413.26,87.73) .. controls (411.19,87.73) and (409.51,86.05) .. (409.51,83.98) -- cycle ;
\draw  [color={rgb, 255:red, 74; green, 144; blue, 226 }  ,draw opacity=1 ][fill={rgb, 255:red, 74; green, 144; blue, 226 }  ,fill opacity=1 ] (551.51,124.98) .. controls (551.51,122.91) and (553.19,121.23) .. (555.26,121.23) .. controls (557.34,121.23) and (559.01,122.91) .. (559.01,124.98) .. controls (559.01,127.05) and (557.34,128.73) .. (555.26,128.73) .. controls (553.19,128.73) and (551.51,127.05) .. (551.51,124.98) -- cycle ;
\draw  [color={rgb, 255:red, 74; green, 144; blue, 226 }  ,draw opacity=1 ][fill={rgb, 255:red, 74; green, 144; blue, 226 }  ,fill opacity=1 ] (363.51,206.98) .. controls (363.51,204.91) and (365.19,203.23) .. (367.26,203.23) .. controls (369.34,203.23) and (371.01,204.91) .. (371.01,206.98) .. controls (371.01,209.05) and (369.34,210.73) .. (367.26,210.73) .. controls (365.19,210.73) and (363.51,209.05) .. (363.51,206.98) -- cycle ;
\draw  [color={rgb, 255:red, 74; green, 144; blue, 226 }  ,draw opacity=1 ][fill={rgb, 255:red, 74; green, 144; blue, 226 }  ,fill opacity=1 ] (151.51,242.98) .. controls (151.51,240.91) and (153.19,239.23) .. (155.26,239.23) .. controls (157.34,239.23) and (159.01,240.91) .. (159.01,242.98) .. controls (159.01,245.05) and (157.34,246.73) .. (155.26,246.73) .. controls (153.19,246.73) and (151.51,245.05) .. (151.51,242.98) -- cycle ;
\draw  [color={rgb, 255:red, 74; green, 144; blue, 226 }  ,draw opacity=1 ][fill={rgb, 255:red, 74; green, 144; blue, 226 }  ,fill opacity=1 ] (256.51,254.98) .. controls (256.51,252.91) and (258.19,251.23) .. (260.26,251.23) .. controls (262.34,251.23) and (264.01,252.91) .. (264.01,254.98) .. controls (264.01,257.05) and (262.34,258.73) .. (260.26,258.73) .. controls (258.19,258.73) and (256.51,257.05) .. (256.51,254.98) -- cycle ;
\draw  [color={rgb, 255:red, 74; green, 144; blue, 226 }  ,draw opacity=1 ][fill={rgb, 255:red, 74; green, 144; blue, 226 }  ,fill opacity=1 ] (331.51,256.98) .. controls (331.51,254.91) and (333.19,253.23) .. (335.26,253.23) .. controls (337.34,253.23) and (339.01,254.91) .. (339.01,256.98) .. controls (339.01,259.05) and (337.34,260.73) .. (335.26,260.73) .. controls (333.19,260.73) and (331.51,259.05) .. (331.51,256.98) -- cycle ;
\draw  [color={rgb, 255:red, 74; green, 144; blue, 226 }  ,draw opacity=1 ][fill={rgb, 255:red, 74; green, 144; blue, 226 }  ,fill opacity=1 ] (474.51,237.98) .. controls (474.51,235.91) and (476.19,234.23) .. (478.26,234.23) .. controls (480.34,234.23) and (482.01,235.91) .. (482.01,237.98) .. controls (482.01,240.05) and (480.34,241.73) .. (478.26,241.73) .. controls (476.19,241.73) and (474.51,240.05) .. (474.51,237.98) -- cycle ;
\draw  [color={rgb, 255:red, 74; green, 144; blue, 226 }  ,draw opacity=1 ][fill={rgb, 255:red, 74; green, 144; blue, 226 }  ,fill opacity=1 ] (593.51,186.98) .. controls (593.51,184.91) and (595.19,183.23) .. (597.26,183.23) .. controls (599.34,183.23) and (601.01,184.91) .. (601.01,186.98) .. controls (601.01,189.05) and (599.34,190.73) .. (597.26,190.73) .. controls (595.19,190.73) and (593.51,189.05) .. (593.51,186.98) -- cycle ;
\draw  [color={rgb, 255:red, 74; green, 144; blue, 226 }  ,draw opacity=1 ][fill={rgb, 255:red, 74; green, 144; blue, 226 }  ,fill opacity=1 ] (473.51,15.98) .. controls (473.51,13.91) and (475.19,12.23) .. (477.26,12.23) .. controls (479.34,12.23) and (481.01,13.91) .. (481.01,15.98) .. controls (481.01,18.05) and (479.34,19.73) .. (477.26,19.73) .. controls (475.19,19.73) and (473.51,18.05) .. (473.51,15.98) -- cycle ;
\draw [color={rgb, 255:red, 74; green, 144; blue, 226 }  ,draw opacity=1 ][line width=1.5]  [dash pattern={on 5.63pt off 4.5pt}]  (174.26,138.98) -- (278.26,139.98) ;
\draw [color={rgb, 255:red, 74; green, 144; blue, 226 }  ,draw opacity=1 ][line width=1.5]  [dash pattern={on 5.63pt off 4.5pt}]  (278.26,139.98) -- (413.26,83.98) ;
\draw [color={rgb, 255:red, 74; green, 144; blue, 226 }  ,draw opacity=1 ][line width=1.5]  [dash pattern={on 5.63pt off 4.5pt}]  (278.26,139.98) -- (367.26,206.98) ;
\draw [color={rgb, 255:red, 74; green, 144; blue, 226 }  ,draw opacity=1 ][line width=1.5]  [dash pattern={on 5.63pt off 4.5pt}]  (260.26,254.98) -- (278.26,139.98) ;
\draw [color={rgb, 255:red, 74; green, 144; blue, 226 }  ,draw opacity=1 ][line width=1.5]  [dash pattern={on 5.63pt off 4.5pt}]  (413.26,83.98) -- (477.26,15.98) ;
\draw [color={rgb, 255:red, 74; green, 144; blue, 226 }  ,draw opacity=1 ][line width=1.5]  [dash pattern={on 5.63pt off 4.5pt}]  (367.26,206.98) -- (335.26,256.98) ;
\draw [color={rgb, 255:red, 74; green, 144; blue, 226 }  ,draw opacity=1 ][line width=1.5]  [dash pattern={on 5.63pt off 4.5pt}]  (155.26,242.98) -- (174.26,138.98) ;
\draw [color={rgb, 255:red, 74; green, 144; blue, 226 }  ,draw opacity=1 ][line width=1.5]  [dash pattern={on 5.63pt off 4.5pt}]  (413.26,83.98) -- (555.26,124.98) ;
\draw [color={rgb, 255:red, 74; green, 144; blue, 226 }  ,draw opacity=1 ][line width=1.5]  [dash pattern={on 5.63pt off 4.5pt}]  (555.26,124.98) -- (597.26,186.98) ;
\draw [color={rgb, 255:red, 74; green, 144; blue, 226 }  ,draw opacity=1 ][line width=1.5]  [dash pattern={on 5.63pt off 4.5pt}]  (81.51,138.98) -- (174.26,138.98) ;
\draw  [color={rgb, 255:red, 74; green, 144; blue, 226 }  ,draw opacity=1 ][fill={rgb, 255:red, 74; green, 144; blue, 226 }  ,fill opacity=1 ] (81.51,138.98) .. controls (81.51,136.91) and (83.19,135.23) .. (85.26,135.23) .. controls (87.34,135.23) and (89.01,136.91) .. (89.01,138.98) .. controls (89.01,141.05) and (87.34,142.73) .. (85.26,142.73) .. controls (83.19,142.73) and (81.51,141.05) .. (81.51,138.98) -- cycle ;
\draw [color={rgb, 255:red, 208; green, 2; blue, 27 }  ,draw opacity=1 ][line width=1.5]  [dash pattern={on 1.69pt off 2.76pt}]  (113.5,165.67) .. controls (76.93,226.29) and (106.5,293.95) .. (314.26,314.98) ;
\draw [color={rgb, 255:red, 74; green, 144; blue, 226 }  ,draw opacity=1 ][line width=1.5]  [dash pattern={on 5.63pt off 4.5pt}]  (278.26,139.98) .. controls (349.52,143.3) and (520.52,174.3) .. (478.26,237.98) ;

\draw (305.51,321.38) node [anchor=north west][inner sep=0.75pt]  [font=\normalsize]  {$z$};
\draw (262,46.39) node [anchor=north west][inner sep=0.75pt]  [font=\normalsize]  {$G$};
\draw (408,94.39) node [anchor=north west][inner sep=0.75pt]  [font=\normalsize,color={rgb, 255:red, 74; green, 144; blue, 226 }  ,opacity=1 ]  {$T^{*}$};
\draw (63,293.6) node [anchor=north west][inner sep=0.75pt]  [color={rgb, 255:red, 208; green, 2; blue, 27 }  ,opacity=1 ]  {$\delta ( z)$};
\end{tikzpicture}
\caption{
The solid edges form the outerplanar graph $G$,
and the dotted edges are the edges incident to the apex node $z$ in $G_z$.
The dashed edges form the dual tree $T^*$.
}
\label{fig:op-dual}
\end{figure}

Note that $\delta(z)=\{(z,v):v\in V\}$ are the edges of a spanning tree of $G_z$, so
$E(G_z)^*\setminus\delta(z)^*$ are the edges of a spanning tree $T^*$ of $G_z^*$.
Each non-leaf node of $T^*$ corresponds to an inner face of $G$, and each leaf of
$T^*$ corresponds to a face of $G_z$ whose boundary contains the apex node $z$.
Also note that we obtain $G^*$ if we combine all the leaves of $T^*$ into a single
node (which would correspond to the outer face of $G$). We will call $T^*$ the dual
tree of the outerplanar graph $G$ (Figure \ref{fig:op-dual}).

Let a central cut of $G$ be a cut $\delta(S)$ such that both of its shores $S$ and
$V\setminus S$ induced  connected subgraphs of $G$. Hence, the shores of a central cut are subpaths of
the outer cycle, so the dual of $\delta(S)$ is a leaf-to-leaf path in $T^*$. Since
the edges of any cut in a connected graph is a disjoint union of central cuts, it suffices to
only consider central cuts.

We want to find a strictly embedded cut-sparsifier $T=(V,F,u^*)$ of $G$ (ie. a spanning
tree $T$ of $G$ with edges weights $u^*$) such that for any nonempty $X\subsetneq V$,
we have
\begin{equation}
\alpha u(\delta_G(X)) \le u^*(\delta_T(X)) \le \beta u(\delta_G(X)) .
\label{cut-sparsifier}
\end{equation}
In the above inequality, we can replace $u^*(\delta_T(X))$ with $u^*(\delta_G(X))$
if we set $u^*(e)=0$ for each edge $e\notin E(T)$. In the dual tree (of $G$),
$\delta_G(X)^*$ is a leaf-to-leaf path for any central cut $\delta(X)$,
so inequality \eqref{cut-sparsifier} is equivalent to
\begin{equation}
\alpha u(P) \le u^*(P) \le \beta u(P)
\label{distance-sparsifier}
\end{equation}
for any leaf-to-leaf path $P$ in $T^*$.

Finally, we give a sufficient property on the weights $u^*$ assigned to the
edges such that all edges of positive weight are in the spanning tree of $G$.
Recall that the dual of the edges not in the spanning tree of $G$ would
form a spanning tree of $G^*$. Since we assign weight 0 to edges not in the
spanning tree of $G$, it is sufficient for the 0 weight edges to form a
spanning subgraph of $G^*$. Since $G^*$ is obtained by combining the leaves
of $T^*$ into a single node, it suffices for each node $v\in V(T^*)$ to
have a 0 weight path from $v$ to a leaf of $T^*$.

\subsection{An algorithm to build a distance-sparsifier of a tree}
\label{sec:non-conservative}

In this section, we present an algorithm to obtain a distance-sparsifier
of a tree. In particular, this allows us to obtain a cut-approximator of
an outerplanar graph from a distance-sparsifier of its dual tree.

Let $T=(V,E,u)$ be a weighted tree where $u:E\to\mathbb{R}^+$ is the
length function on $T$. Let $L\subset V$ be the leaves of $T$. We assign
non-negative weights $u^*$ to the edges of $T$. Let $d$ be the shortest
path metric induced by the original weights $u$, and let $d^*$ be the
shortest path metric induced by the new weights $u^*$. We want the following
two conditions to hold:
\begin{enumerate}
    \item there exists a 0 weight path from each $v\in V$ to a leaf of $T$.
    \item for any two leaves $x,y\in L$, we have
    \begin{equation}
    \frac14 d(x,y) \le d^*(x,y) \le 2 d(x,y) .
    \label{tree-bounds}
    \end{equation}
\end{enumerate}

We define $u^*$ recursively as follows. Let $r$ be a non-leaf node of $T$
(we are done if no such nodes exist), and consider $T$ to be rooted at
$r$. For $v\in V$, let $T(v)$ denote the subtree rooted at $v$, and let $h(v)$
denote the \emph{height} of $v$, defined by $h(v)=\min\{d(v,x):x\in L\cap T(v)\}$. Now,
let $r_1, ..., r_k$ be the points in $T$ that are at distance exactly $h(r)/2$
from $r$. Without loss of generality, suppose that each $r_i$ is a node
(otherwise we can subdivide the edge to get a node), and order the $r_i$'s
by increasing $h(r_i)$, that is $h(r_{i-1})\le h(r_i)$ for each $i=2,...,k$.
Furthermore, suppose that we have already assigned weights to the edges in
each subtree $T(r_i)$ using this algorithm, so it remains to assign weights
to the edges not in any of these subtrees. We assign a weight of $h(r_i)$ to
the first edge on the path from $r_i$ to $r$ for each $i=2,...,k$, and weight
0 to all other edges (Figure \ref{fig:algorithm}).
In particular, all edges on the path from $r_1$ to $r$ receive weight $0$.
This algorithm terminates because the length of the
longest path from the root to a leaf decreases by at least half the length
of the shortest edge incident to a leaf in each iteration.


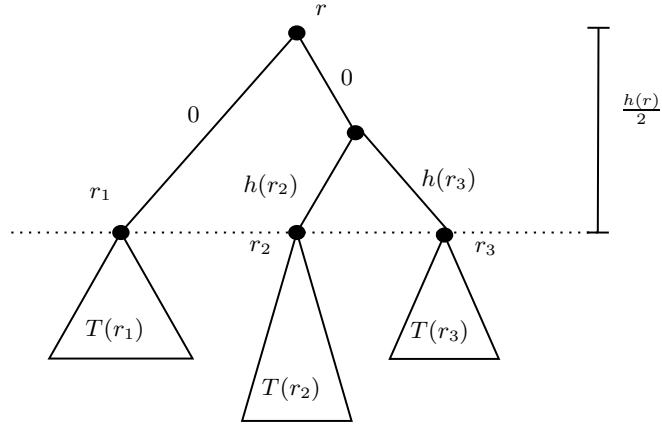
\begin{figure}
\centering
\begin{tikzpicture}[x=0.75pt,y=0.75pt,yscale=-1,xscale=1]

\draw  [fill={rgb, 255:red, 0; green, 0; blue, 0 }  ,fill opacity=1 ] (145.38,47.73) .. controls (145.38,45.9) and (147.12,44.42) .. (149.27,44.42) .. controls (151.41,44.42) and (153.15,45.9) .. (153.15,47.73) .. controls (153.15,49.56) and (151.41,51.04) .. (149.27,51.04) .. controls (147.12,51.04) and (145.38,49.56) .. (145.38,47.73) -- cycle ;
\draw  [dash pattern={on 0.84pt off 2.51pt}]  (5.26,148.48) -- (301.05,148.48) ;
\draw  [fill={rgb, 255:red, 0; green, 0; blue, 0 }  ,fill opacity=1 ] (56.64,148.48) .. controls (56.64,146.65) and (58.38,145.16) .. (60.53,145.16) .. controls (62.68,145.16) and (64.42,146.65) .. (64.42,148.48) .. controls (64.42,150.31) and (62.68,151.79) .. (60.53,151.79) .. controls (58.38,151.79) and (56.64,150.31) .. (56.64,148.48) -- cycle ;
\draw  [fill={rgb, 255:red, 0; green, 0; blue, 0 }  ,fill opacity=1 ] (219.87,149.74) .. controls (219.87,147.91) and (221.61,146.42) .. (223.76,146.42) .. controls (225.91,146.42) and (227.65,147.91) .. (227.65,149.74) .. controls (227.65,151.56) and (225.91,153.05) .. (223.76,153.05) .. controls (221.61,153.05) and (219.87,151.56) .. (219.87,149.74) -- cycle ;
\draw  [fill={rgb, 255:red, 0; green, 0; blue, 0 }  ,fill opacity=1 ] (145.38,148.48) .. controls (145.38,146.65) and (147.12,145.16) .. (149.27,145.16) .. controls (151.41,145.16) and (153.15,146.65) .. (153.15,148.48) .. controls (153.15,150.31) and (151.41,151.79) .. (149.27,151.79) .. controls (147.12,151.79) and (145.38,150.31) .. (145.38,148.48) -- cycle ;
\draw  [fill={rgb, 255:red, 0; green, 0; blue, 0 }  ,fill opacity=1 ] (174.96,98.1) .. controls (174.96,96.28) and (176.7,94.79) .. (178.84,94.79) .. controls (180.99,94.79) and (182.73,96.28) .. (182.73,98.1) .. controls (182.73,99.93) and (180.99,101.42) .. (178.84,101.42) .. controls (176.7,101.42) and (174.96,99.93) .. (174.96,98.1) -- cycle ;
\draw    (149.27,47.73) -- (178.84,98.1) ;
\draw    (60.53,148.48) -- (149.27,47.73) ;
\draw    (150.74,145.96) -- (180.32,95.59) ;
\draw    (225.24,147.22) -- (180.32,95.59) ;
\draw   (60.53,148.48) -- (96.68,212.09) -- (24.38,212.09) -- cycle ;
\draw   (149.27,148.48) -- (176.87,243.57) -- (121.66,243.57) -- cycle ;
\draw   (223.65,149.74) -- (251.15,212.27) -- (196.15,212.27) -- cycle ;
\draw    (301.41,45.14) -- (301.05,148.48) ;
\draw [shift={(301.05,148.48)}, rotate = 270.2] [color={rgb, 255:red, 0; green, 0; blue, 0 }  ][line width=0.75]    (0,5.59) -- (0,-5.59)   ;
\draw [shift={(301.41,45.14)}, rotate = 270.2] [color={rgb, 255:red, 0; green, 0; blue, 0 }  ][line width=0.75]    (0,5.59) -- (0,-5.59)   ;

\draw (43.21,124.3) node [anchor=north west][inner sep=0.75pt]    {$r_{1}$};
\draw (123.94,151.24) node [anchor=north west][inner sep=0.75pt]    {$r_{2}$};
\draw (237.48,151.79) node [anchor=north west][inner sep=0.75pt]    {$r_{3}$};
\draw (157.17,32.13) node [anchor=north west][inner sep=0.75pt]    {$r$};
\draw (310.76,76.38) node [anchor=north west][inner sep=0.75pt]    {$\frac{h( r)}{2}$};
\draw (120.92,117.48) node [anchor=north west][inner sep=0.75pt]  [font=\small]  {$h( r_{2})$};
\draw (210.64,114.48) node [anchor=north west][inner sep=0.75pt]  [font=\small]  {$h( r_{3})$};
\draw (169.88,64.56) node [anchor=north west][inner sep=0.75pt]    {$0$};
\draw (92.65,83.7) node [anchor=north west][inner sep=0.75pt]    {$0$};
\draw (41.25,190.2) node [anchor=north west][inner sep=0.75pt]    {$T( r_{1})$};
\draw (130.19,220.15) node [anchor=north west][inner sep=0.75pt]    {$T( r_{2})$};
\draw (205.09,191.39) node [anchor=north west][inner sep=0.75pt]    {$T( r_{3})$};
\end{tikzpicture}
\caption{
The algorithm assigns weights to the edges above $r_1,...,r_k$,
and is run recursively on the subtrees $T(r_1),...,T(r_k)$.
}
\label{fig:algorithm}
\end{figure}

Since we assign 0 weight to edges on the $r_1r$ path,
Condition 1 is satisfied for all nodes above the $r_i$'s in the tree by construction. It remains to prove Condition 2.
We  use the following upper and lower bounds. For each leaf $x\in L$,
\begin{align}
    d^*(x,r) &\le 2d(x,r) - h(r) \label{upper-bound} , \\
    d^*(x,r) &\ge d(x,r) - h(r) \label{lower-bound} .
\end{align}

We prove the upper bound in \eqref{upper-bound} by induction. We are
done if $T$ only has 0 weight edges, and the cases that cause the algorithm
to terminate will only have 0 weight edges. For the induction, we consider
two separate cases depending on whether $x\in T(r_1)$.

\textbf{Case 1}: $x\in T(r_1)$.
\begin{align*}
d^*(x, r)
&= d^*(x, r_1) + d^*(r_1, r)
&& \textrm{($r_1$ is between $x$ and $r$)} \\
&= d^*(x, r_1)
&& \textrm{(by definition of $u^*$)} \\
&\le 2d(x, r_1) - h(r_1)
&& \textrm{(by induction)} \\
&= 2d(x, r) - 2d(r, r_1) - h(r_1)
&& \textrm{($r_1$ is between $x$ and $r$)} \\
&= 2d(x, r) - \frac32 h(r)
&& \textrm{($h(r_1)=h(r)/2$ by definition of $r_1$)} \\
&\le 2d(x, r) - h(r)
\end{align*}

\textbf{Case 2}: $x\in T(r_i)$ for some $i\neq1$.
\begin{align*}
d^*(x, r)
&= d^*(x, r_i) + d^*(r_i, r)
&& \textrm{($r_i$ is between $x$ and $r$)} \\
&= d^*(x, r_i) + h(r_i)
&& \textrm{(by definition of $u^*$)} \\
&\le 2d(x, r_i) - h(r_i) + h(r_i)
&& \textrm{(by induction)} \\
&= 2d(x, r) - 2d(r_i, r)
&& \textrm{($r_i$ is between $x$ and $r$)} \\
&= 2d(x, r) - h(r)
&& \textrm{($d(r_i, r) = h(r)/2$ by definition of $r_i$)}
\end{align*}

This proves inequality \eqref{upper-bound}.

We prove the lower bound in \eqref{lower-bound} similarly.

\textbf{Case 1}: $x \in T(r_1)$.
\begin{align*}
d^*(x, r)
&= d^*(x, r_1) + d^*(r_1, r)
&& \textrm{($r_1$ is between $x$ and $r$)} \\
&= d^*(x, r_1)
&& \textrm{(by definition of $u^*$)} \\
&\ge d(x, r_1) - h(r_1)
&& \textrm{(by induction)} \\
&= d(x, r) - d(r, r_1) - h(r_1)
&& \textrm{($r_1$ is between $x$ and $r$)} \\
&= d(x, r) - h(r)
&& \textrm{(by definition of $r_1$)}
\end{align*}

\textbf{Case 2}: $x \in T(r_i)$ for some $i\neq1$.
\begin{align*}
d^*(x, r)
&= d^*(x, r_i) + d^*(r_i, r)
&& \textrm{($r_i$ is between $x$ and $r$)} \\
&= d^*(x, r_i) + h(r_i)
&& \textrm{(by definition of $u^*$)} \\
&\ge d(x, r_i) - h(r_i) + h(r_i)
&& \textrm{(by induction)} \\
&= d(x, r) - d(r_i, r)
&& \textrm{($r_i$ is between $x$ and $r$)} \\
&= d(x, r) - h(r)/2
&& \textrm{($d(r_i, r) = h(r)/2$ by definition of $r_i$)} \\
&\ge d(x, r) - h(r)
\end{align*}

This proves inequality \eqref{lower-bound}

Finally, we prove property 2, that is inequality \eqref{tree-bounds},
by induction. Let $x,y\in L$ be two leaves of $T$. Suppose that
$x\in T(r_i)$ and $y\in T(r_j)$. By induction, we may assume that
$i\neq j$, so without loss of generality, suppose that $i<j$.

We prove the upper bound.
\begin{align*}
d^*(x, y) &= d^*(x, r_i) + d^*(r_i, r_j) + d^*(r_j, y) \\
&\le 2d(x, r_i) - h(r_i)  + 2d(y, r_j) - h(r_j) + d^*(r_i, r_j)
&& \textrm{(by \eqref{upper-bound})} \\
&\le 2d(x, r_i) - h(r_i)  + 2d(y, r_j) - h(r_j) + h(r_i) + h(r_j)
&& \textrm{(by definition of $u^*$)} \\
&= 2d(x, r_i) + 2d(y, r_j) \\
&\le 2d(x, y)
\end{align*}

We prove the lower bound.
\begin{align*}
d(x, y)
&= d(x, r_i) + d(r_i, r_j) + d(r_j, y) \\
&\le d(x, r_i) + d(r_j, y) + h(r_i) + h(r_j) && \\
&\qquad\qquad \textrm{(because $d(r, r_i)=h(r)/2\le h(r_i)$ for all $i\in[k]$)} && \\
&\le 2d(x, r_i) + 2d(r_j, y)
&& \textrm{(by definition of $h$)} \\
&\le 2d^*(x, r_i) + 2h(r_i) + 2d^*(y, r_j) + 2h(r_j)
&& \textrm{(by \eqref{lower-bound})} \\
&= 2d^*(x, y) - 2d^*(r_i, r_j) + 2h(r_i) + 2h(r_j) .
\end{align*}
Now we finish the proof of the lower bound by considering two cases.

\textbf{Case 1}: $i = 1$, that is $x$ is in the first subtree.
\begin{align*}
d(x, y)
&\le 2d^*(x, y) - 2d^*(r_1, r_j) + 2h(r_1) + 2h(r_j) \\
&= 2d^*(x, y) - 2h(r_j) + 2h(r_1) + 2h(r_j)
&& \textrm{(by definition of $u^*$)} \\
&\le 2d^*(x, y) + 2h(r_1) \\
&\le 4d^*(x, y)
\end{align*}

\textbf{Case 2}: $i > 1$, that is neither $x$ nor $y$ is in the first subtree.
\begin{align*}
d(x, y)
&\le 2d^*(x, y) - 2d^*(r_i, r_j) + 2h(r_i) + 2h(r_j) \\
&= 2d^*(x, y) - 2h(r_i) - 2h(r_j) + 2h(r_i) + 2h(r_j)
&& \textrm{(by definition of $u^*$)} \\
&= 2d^*(x, y)
\end{align*}

This completes the proof of property 2.

\section{Maximum Weight Disjoint Paths} 

In this section we  prove our main result for {\sc edp}, Theorem~\ref{thm:edp}.

\subsection{Required Elements}
\label{sec:required}

We first prove the following result  which establishes conditions for when 
the cut condition implies routability.  

\OProute*

The novelty in this statement is that we do not require the Eulerian condition
on $G+H$. This condition is needed in virtually all classical results for edge-disjoint paths. In fact, even when $G$ is a $4$-cycle and $H$ consists of a matching of size $2$, the cut condition need not be sufficient to guarantee routability. The main exception is the case when $G$ is a tree and a trivial greedy algorithm suffices to route $H$. We prove the theorem by giving a simple (but not so simple) algorithm to compute a  routing.

To prove this theorem, we  need the following $2$-node reduction lemma which is generally known.
\begin{lemma}
\label{lemma:2con-cc}
Let $G$ be a graph and let $H$ be a collection of demands that satisfies the cut condition.
Let $G_1,...,G_k$ be the blocks of $G$ (the 2-node connected components and the cut edges (aka bridges) of $G$).
Let $H_i$ be the collection of nontrivial (i.e., non-loop) demands after contracting
each edge $e\in E(G)\setminus E(G_i)$.
Then each $G_i,H_i$ satisfies the cut condition.
Furthermore, if $G$ (or $G+H$) was outerplanar (or planar),
then each $G_i$ (resp. $G_i+H_i$) is outerplanar (resp. planar).
Moreover, if each $H_i$ is routable in $G_i$, then $H$ is routable in $G$.
\end{lemma}

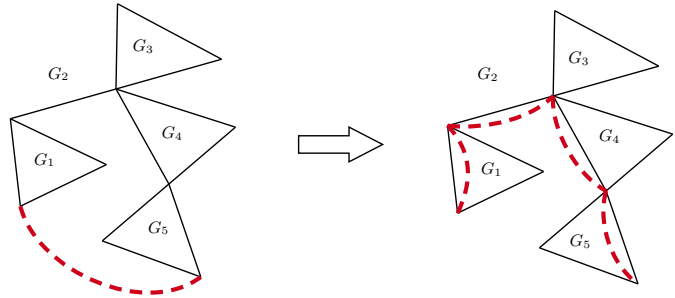
\begin{figure}[htbp]
\centering
\scalebox{0.7}{
\begin{tikzpicture}[x=0.75pt,y=0.75pt,yscale=-1,xscale=1]

\draw    (120.4,91.54) -- (196.52,70.07) ;
\draw    (44.29,113) -- (51.52,176.07) ;
\draw    (44.29,113) -- (120.4,91.54) ;
\draw    (120.4,91.54) -- (121.52,30.07) ;
\draw    (196.52,70.07) -- (121.52,30.07) ;
\draw    (51.52,176.07) -- (113.52,146.07) ;
\draw    (158.52,160.07) -- (206.52,119.07) ;
\draw    (120.4,91.54) -- (206.52,119.07) ;
\draw    (120.4,91.54) -- (158.52,160.07) ;
\draw    (44.29,113) -- (113.52,146.07) ;
\draw   (252,125.27) -- (288.31,125.27) -- (288.31,119) -- (312.52,131.54) -- (288.31,144.07) -- (288.31,137.81) -- (252,137.81) -- cycle ;
\draw    (158.52,160.07) -- (181.52,227.07) ;
\draw    (181.52,227.07) -- (110.52,201.07) ;
\draw    (110.52,201.07) -- (158.52,160.07) ;
\draw [color={rgb, 255:red, 208; green, 2; blue, 27 }  ,draw opacity=1 ][line width=2.25]  [dash pattern={on 6.75pt off 4.5pt}]  (51.52,176.07) .. controls (63.52,223.07) and (141.52,257.07) .. (181.52,227.07) ;
\draw    (435.4,96.54) -- (511.52,75.07) ;
\draw    (359.29,118) -- (366.52,181.07) ;
\draw    (359.29,118) -- (435.4,96.54) ;
\draw    (435.4,96.54) -- (436.52,35.07) ;
\draw    (511.52,75.07) -- (436.52,35.07) ;
\draw    (366.52,181.07) -- (428.52,151.07) ;
\draw    (473.52,165.07) -- (521.52,124.07) ;
\draw    (435.4,96.54) -- (521.52,124.07) ;
\draw    (435.4,96.54) -- (473.52,165.07) ;
\draw    (359.29,118) -- (428.52,151.07) ;
\draw    (473.52,165.07) -- (496.52,232.07) ;
\draw    (496.52,232.07) -- (425.52,206.07) ;
\draw    (425.52,206.07) -- (473.52,165.07) ;
\draw [color={rgb, 255:red, 208; green, 2; blue, 27 }  ,draw opacity=1 ][line width=2.25]  [dash pattern={on 6.75pt off 4.5pt}]  (366.52,181.07) .. controls (376.52,160.07) and (378.52,138.07) .. (359.29,118) ;
\draw [color={rgb, 255:red, 208; green, 2; blue, 27 }  ,draw opacity=1 ][line width=2.25]  [dash pattern={on 6.75pt off 4.5pt}]  (473.52,165.07) .. controls (467.52,186.07) and (475.52,217.07) .. (496.52,232.07) ;
\draw [color={rgb, 255:red, 208; green, 2; blue, 27 }  ,draw opacity=1 ][line width=2.25]  [dash pattern={on 6.75pt off 4.5pt}]  (435.4,96.54) .. controls (433.52,116.07) and (448.52,151.07) .. (473.52,165.07) ;
\draw [color={rgb, 255:red, 208; green, 2; blue, 27 }  ,draw opacity=1 ][line width=2.25]  [dash pattern={on 6.75pt off 4.5pt}]  (359.29,118) .. controls (389.52,121.07) and (420.52,112.07) .. (435.4,96.54) ;

\draw (60,136.4) node [anchor=north west][inner sep=0.75pt]    {$G_{1}$};
\draw (70,74.4) node [anchor=north west][inner sep=0.75pt]    {$G_{2}$};
\draw (131,54.4) node [anchor=north west][inner sep=0.75pt]    {$G_{3}$};
\draw (152,117.4) node [anchor=north west][inner sep=0.75pt]    {$G_{4}$};
\draw (142,185.4) node [anchor=north west][inner sep=0.75pt]    {$G_{5}$};
\draw (467,119.4) node [anchor=north west][inner sep=0.75pt]    {$G_{4}$};
\draw (446,195.4) node [anchor=north west][inner sep=0.75pt]    {$G_{5}$};
\draw (445,63.4) node [anchor=north west][inner sep=0.75pt]    {$G_{3}$};
\draw (379.35,74.67) node [anchor=north west][inner sep=0.75pt]    {$G_{2}$};
\draw (382,143.4) node [anchor=north west][inner sep=0.75pt]    {$G_{1}$};
\end{tikzpicture}
}
\caption{
The new demand edges that replace a demand edge whose terminals belong in different blocks.
Solid edges represent edges of $G$ and dashed edges represent demand edges.
}
\label{fig:route-contract}
\end{figure}

\begin{proof}
Consider the edge contractions to be done on $G+H$ to obtain $G_i+H_i$.
Then, any cut in $G_i+H_i$ was also a cut in $G+H$.
Since $G,H$ satisfies the cut condition, then $G_i,H_i$ must also satisfy the cut condition.
Furthermore, edge contraction preserves planarity and outerplanarity.

For each $st \in H$ and each $G_i$, the reduction process produces
a request $s_it_i$ in $G_i$. If this is not a loop, then $s_i,t_i$ lie in different components of $G$ after deleting the edges of $G_i$. In this case,
we say that $st$ {\em spawns} $s_it_i$.  Let $J$ be the set of edges  spawned by a demand $st$.
It is easy to see that the edges of  $J$ form an $st$ path.  
Hence if each $H_i$ is routable in $G_i$, we have that $H$ is routable in $G$.
\end{proof}

\begin{proof}[Proof of theorem \ref{thm:OP}]
Without loss of generality, we may assume that the edges of $G$ (resp. $H$) have unit capacity (resp. demand).
Otherwise, we may place $u(e)$  (resp. $d(e)$) parallel copies of such an edge $e$. 
In the algorithmic proof, we may also assume that $G$ is 2-node connected.
Otherwise, we may apply Lemma \ref{lemma:2con-cc} and consider each 2-node
connected component of $G$ separately.
When working with  2-node connected $G$, the boundary of its outer face is a simple cycle.
 So we label the nodes $v_1,...,v_n$
by the order they appear on this cycle.

\begin{figure}
\centering
\begin{tikzpicture}[x=0.75pt,y=0.75pt,yscale=-1,xscale=1]

\draw   (43,123.26) .. controls (43,69.54) and (86.54,26) .. (140.26,26) .. controls (193.97,26) and (237.52,69.54) .. (237.52,123.26) .. controls (237.52,176.97) and (193.97,220.52) .. (140.26,220.52) .. controls (86.54,220.52) and (43,176.97) .. (43,123.26) -- cycle ;
\draw  [fill={rgb, 255:red, 0; green, 0; blue, 0 }  ,fill opacity=1 ] (120,218.76) .. controls (120,216.68) and (121.68,215) .. (123.76,215) .. controls (125.84,215) and (127.52,216.68) .. (127.52,218.76) .. controls (127.52,220.84) and (125.84,222.52) .. (123.76,222.52) .. controls (121.68,222.52) and (120,220.84) .. (120,218.76) -- cycle ;
\draw  [fill={rgb, 255:red, 0; green, 0; blue, 0 }  ,fill opacity=1 ] (165,215.76) .. controls (165,213.68) and (166.68,212) .. (168.76,212) .. controls (170.84,212) and (172.52,213.68) .. (172.52,215.76) .. controls (172.52,217.84) and (170.84,219.52) .. (168.76,219.52) .. controls (166.68,219.52) and (165,217.84) .. (165,215.76) -- cycle ;
\draw  [fill={rgb, 255:red, 0; green, 0; blue, 0 }  ,fill opacity=1 ] (208,189.76) .. controls (208,187.68) and (209.68,186) .. (211.76,186) .. controls (213.84,186) and (215.52,187.68) .. (215.52,189.76) .. controls (215.52,191.84) and (213.84,193.52) .. (211.76,193.52) .. controls (209.68,193.52) and (208,191.84) .. (208,189.76) -- cycle ;
\draw  [fill={rgb, 255:red, 0; green, 0; blue, 0 }  ,fill opacity=1 ] (228,155.76) .. controls (228,153.68) and (229.68,152) .. (231.76,152) .. controls (233.84,152) and (235.52,153.68) .. (235.52,155.76) .. controls (235.52,157.84) and (233.84,159.52) .. (231.76,159.52) .. controls (229.68,159.52) and (228,157.84) .. (228,155.76) -- cycle ;
\draw  [fill={rgb, 255:red, 0; green, 0; blue, 0 }  ,fill opacity=1 ] (77,200.76) .. controls (77,198.68) and (78.68,197) .. (80.76,197) .. controls (82.84,197) and (84.52,198.68) .. (84.52,200.76) .. controls (84.52,202.84) and (82.84,204.52) .. (80.76,204.52) .. controls (78.68,204.52) and (77,202.84) .. (77,200.76) -- cycle ;
\draw  [fill={rgb, 255:red, 0; green, 0; blue, 0 }  ,fill opacity=1 ] (45,156.76) .. controls (45,154.68) and (46.68,153) .. (48.76,153) .. controls (50.84,153) and (52.52,154.68) .. (52.52,156.76) .. controls (52.52,158.84) and (50.84,160.52) .. (48.76,160.52) .. controls (46.68,160.52) and (45,158.84) .. (45,156.76) -- cycle ;
\draw  [fill={rgb, 255:red, 0; green, 0; blue, 0 }  ,fill opacity=1 ] (41,108.76) .. controls (41,106.68) and (42.68,105) .. (44.76,105) .. controls (46.84,105) and (48.52,106.68) .. (48.52,108.76) .. controls (48.52,110.84) and (46.84,112.52) .. (44.76,112.52) .. controls (42.68,112.52) and (41,110.84) .. (41,108.76) -- cycle ;
\draw  [fill={rgb, 255:red, 0; green, 0; blue, 0 }  ,fill opacity=1 ] (100,33.76) .. controls (100,31.68) and (101.68,30) .. (103.76,30) .. controls (105.84,30) and (107.52,31.68) .. (107.52,33.76) .. controls (107.52,35.84) and (105.84,37.52) .. (103.76,37.52) .. controls (101.68,37.52) and (100,35.84) .. (100,33.76) -- cycle ;
\draw  [fill={rgb, 255:red, 0; green, 0; blue, 0 }  ,fill opacity=1 ] (219,74.76) .. controls (219,72.68) and (220.68,71) .. (222.76,71) .. controls (224.84,71) and (226.52,72.68) .. (226.52,74.76) .. controls (226.52,76.84) and (224.84,78.52) .. (222.76,78.52) .. controls (220.68,78.52) and (219,76.84) .. (219,74.76) -- cycle ;
\draw    (48.76,156.76) .. controls (96.52,140.07) and (183.52,100.07) .. (222.76,74.76) ;
\draw    (168.76,215.76) .. controls (139.52,170.07) and (89.52,152.07) .. (48.76,156.76) ;
\draw [color={rgb, 255:red, 208; green, 2; blue, 27 }  ,draw opacity=1 ] [dash pattern={on 4.5pt off 4.5pt}]  (48.76,156.76) .. controls (120.52,144.07) and (165.52,151.07) .. (211.76,189.76) ;
\draw [color={rgb, 255:red, 208; green, 2; blue, 27 }  ,draw opacity=1 ] [dash pattern={on 4.5pt off 4.5pt}]  (44.76,108.76) .. controls (89.52,106.07) and (155.52,96.07) .. (222.76,74.76) ;
\draw [color={rgb, 255:red, 208; green, 2; blue, 27 }  ,draw opacity=1 ][line width=3]  [dash pattern={on 7.88pt off 4.5pt}]  (103.76,33.76) .. controls (141.52,62.07) and (173.52,70.07) .. (222.76,74.76) ;
\draw [color={rgb, 255:red, 208; green, 2; blue, 27 }  ,draw opacity=1 ] [dash pattern={on 4.5pt off 4.5pt}]  (222.76,74.76) .. controls (194.52,119.07) and (198.52,152.07) .. (211.76,189.76) ;
\draw  [fill={rgb, 255:red, 0; green, 0; blue, 0 }  ,fill opacity=1 ] (60,63.76) .. controls (60,61.68) and (61.68,60) .. (63.76,60) .. controls (65.84,60) and (67.52,61.68) .. (67.52,63.76) .. controls (67.52,65.84) and (65.84,67.52) .. (63.76,67.52) .. controls (61.68,67.52) and (60,65.84) .. (60,63.76) -- cycle ;
\draw  [fill={rgb, 255:red, 0; green, 0; blue, 0 }  ,fill opacity=1 ] (155,28.76) .. controls (155,26.68) and (156.68,25) .. (158.76,25) .. controls (160.84,25) and (162.52,26.68) .. (162.52,28.76) .. controls (162.52,30.84) and (160.84,32.52) .. (158.76,32.52) .. controls (156.68,32.52) and (155,30.84) .. (155,28.76) -- cycle ;
\draw  [fill={rgb, 255:red, 0; green, 0; blue, 0 }  ,fill opacity=1 ] (194,44.76) .. controls (194,42.68) and (195.68,41) .. (197.76,41) .. controls (199.84,41) and (201.52,42.68) .. (201.52,44.76) .. controls (201.52,46.84) and (199.84,48.52) .. (197.76,48.52) .. controls (195.68,48.52) and (194,46.84) .. (194,44.76) -- cycle ;
\draw  [fill={rgb, 255:red, 0; green, 0; blue, 0 }  ,fill opacity=1 ] (233.76,119.5) .. controls (233.76,117.42) and (235.44,115.74) .. (237.52,115.74) .. controls (239.59,115.74) and (241.28,117.42) .. (241.28,119.5) .. controls (241.28,121.58) and (239.59,123.26) .. (237.52,123.26) .. controls (235.44,123.26) and (233.76,121.58) .. (233.76,119.5) -- cycle ;

\draw (113,230.4) node [anchor=north west][inner sep=0.75pt]    {$v_{1}$};
\draw (57,207.4) node [anchor=north west][inner sep=0.75pt]    {$v_{2}$};
\draw (170,228.4) node [anchor=north west][inner sep=0.75pt]    {$v_{n}$};
\draw (211,205.4) node [anchor=north west][inner sep=0.75pt]    {$v_{n-1}$};
\draw (247.01,161) node [anchor=north west][inner sep=0.75pt]  [rotate=-19.44]  {$\vdots $};
\draw (23.19,177.54) node [anchor=north west][inner sep=0.75pt]  [rotate=-325.57]  {$\vdots $};
\draw (80,7.4) node [anchor=north west][inner sep=0.75pt]    {$v_{i}$};
\draw (233,54.4) node [anchor=north west][inner sep=0.75pt]    {$v_{j}$};
\end{tikzpicture}

\caption{
The solid edges form the outerplanar graph $G$.
The dashed edges are the demand edges.
The thick dashed edge is a valid edge to route
because there are no terminals $v_k$ with $i<k<j$.
}
\label{fig:route-op}
\end{figure}
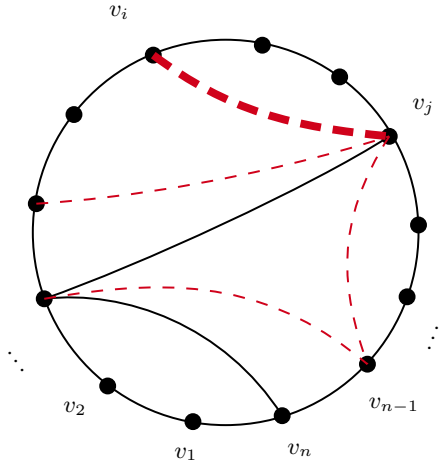

If there are no demand edges, then we are done. Otherwise, since $G+H$ is
outerplanar, without loss of generality there exists $i<j$ such that  $v_iv_j\in E(H)$ and no
$v_k$ is a terminal for $i<k<j$ (Figure \ref{fig:route-op}).
Consider the outer face path $P=v_i,v_{i+1},...,v_j$.
We show that the cut condition is still satisfied after removing both the path $P$
and the demand $v_iv_j$. This represents routing the demand $v_iv_j$ along the path $P$.

Consider a central cut $\delta_G(X)$.
Suppose that $v_i$ and $v_j$ are on opposite sides of the cut. Then, we decrease
both $\delta_G(X)$ and $\delta_H(X)$ by 1, so the cut condition holds.
Suppose that $v_i,v_j\notin X$, that is $v_i$ and $v_j$ are on the same side of the cut.
Then, either $X\subset V(P)\setminus\{v_i,v_j\}$ or $X\cap V(P)=\emptyset$.
We are done if $X\cap V(P)=\emptyset$ because $\delta_G(X)\cap E(P)=0$.
Otherwise, $X\subset V(P)\setminus\{v_i,v_j\}$ contains no terminals,
so we cannot violate the cut condition.
\end{proof}

We  also need the following result from \cite{chekuri2007multicommodity}.
\begin{theorem}
\label{thm:cms}
Let $T$ be a tree with integer edge capacities $u(e)$. Let $H$ denote a demand graph such that  each fundamental cut of $H$ induced by an edge  $e \in T$ contains
at most $k u(e)$ edges of $H$.  We may then partition $H$ into at most $4k$ edges sets $H_1, \ldots ,H_{4k}$ such that each $H_i$ is routable in $T$.
\end{theorem}

\subsection{Proof of the Main Theorem}

\outerplanarWEDPapprox*

\begin{proof}
We first run the algorithms to produce a integer-capacitated tree $T,\hat{u}$ which is an $O(1)$ cut approximator for $G$. In addition $T$ is a subtree and it is a conservative approximator for each cut in $G$. First,  we prove that the maximum weight routable in $T$ is not too much smaller than for $G$ (in either the {\sc edp} or {\sc anf} model).  To see this let $S$ be an optimal solution in $G$, whose value is {\sc opt(G)}. Clearly $S$ satisfies the cut condition in $G$ and hence by Theorem~\ref{thm:tree} it satisfies $14 \cdot$ the cut condition in $T,\hat{u}$.
Thus by Theorem~\ref{thm:cms}  there are $56$ sets such that $S = \cup_{i=1}^{56} S_i$ and each $S_i$ is routable in $T$. Hence one of the
sets $S_i$ accrues at least $\frac{1}{56}^{th}$ the profit from {\sc opt(G)}.

Now we use the factor $4$ approximation \cite{chekuri2007multicommodity} to solve the maximum {\sc edp=anf} problem for $T,\hat{u}$. Let $S$ be a subset of requests which are routable in $T$ and have weight at least $\frac{1}{4}$ {\sc opt(T)} $ \geq \frac{1}{224}$ {\sc opt(G)}.  Since $T$ is a subtree of $G$ we have that $G+T$ is outerplanar. Since $T,\hat{u}$ is an under-estimator of cuts in $G$, we have that the edges of  $T$ (viewed as requests) satisfies the cut condition in $G$. Hence by Theorem~\ref{thm:OP} we may route these single edge requests in $G$. Hence since $S$ can route in $T$, we have that $S$ can also route in $G$, completing the proof.
\end{proof}


\section{Conclusions}

The technique of finding a single-tree constant-factor cut approximator (for a global constant) appears to hit a limit at outerplanar graphs. 
 It would  be
interesting to find a graph parameter $k$ which ensures a single-tree $O(f(k))$ cut approximator.

\noindent
The authors thank Nick Harvey for his valuable feedback on this article.

\bibliographystyle{plain}
\bibliography{references}

\begin{thebibliography}{10}

\bibitem{AbrahamBN08}
Ittai Abraham, Yair Bartal, and Ofer Neiman.
\newblock Nearly tight low stretch spanning trees.
\newblock In {\em FOCS}, pages 781--790, 2008.

\bibitem{alon1995graph}
Noga Alon, Richard~M Karp, David Peleg, and Douglas West.
\newblock A graph-theoretic game and its application to the k-server problem.
\newblock {\em SIAM Journal on Computing}, 24(1):78--100, 1995.

\bibitem{andersen2009interchanging}
Reid Andersen and Uriel Feige.
\newblock Interchanging distance and capacity in probabilistic mappings.
\newblock {\em arXiv preprint arXiv:0907.3631}, 2009.

\bibitem{CKS-planar-constant}
C.~Chekuri, S.~Khanna, and F.B. Shepherd.
\newblock Edge-disjoint paths in planar graphs with constant congestion.
\newblock {\em SIAM Journal on Computing}, 39:281--301, 2009.

\bibitem{chekuri2013poly}
Chandra Chekuri and Alina Ene.
\newblock Poly-logarithmic approximation for maximum node disjoint paths with
  constant congestion.
\newblock In {\em Proceedings of the twenty-fourth annual ACM-SIAM symposium on
  Discrete algorithms}, pages 326--341. SIAM, 2013.

\bibitem{Chekuri04a}
Chandra Chekuri, Sanjeev Khanna, and F.~Bruce Shepherd.
\newblock The all-or-nothing multicommodity flow problem.
\newblock In {\em STOC '04: Proceedings of the thirty-sixth annual ACM
  symposium on Theory of computing}, pages 156--165, New York, NY, USA, 2004.
  ACM.

\bibitem{chekuri2013all}
Chandra Chekuri, Sanjeev Khanna, and F~Bruce Shepherd.
\newblock The all-or-nothing multicommodity flow problem.
\newblock {\em SIAM Journal on Computing}, 42(4):1467--1493, 2013.

\bibitem{chekuri2007multicommodity}
Chandra Chekuri, Marcelo Mydlarz, and F~Bruce Shepherd.
\newblock Multicommodity demand flow in a tree and packing integer programs.
\newblock {\em ACM Transactions on Algorithms (TALG)}, 3(3):27--es, 2007.

\bibitem{chekuri2013maximum}
Chandra Chekuri, Guyslain Naves, and F~Bruce Shepherd.
\newblock Maximum edge-disjoint paths in k-sums of graphs.
\newblock In {\em International Colloquium on Automata, Languages, and
  Programming}, pages 328--339. Springer, 2013.

\bibitem{chuzhoy2012polylogarithimic}
J.~Chuzhoy and S.~Li.
\newblock A polylogarithimic approximation algorithm for edge-disjoint paths
  with congestion 2.
\newblock {\em arXiv preprint arXiv:1208.1272}, 2012.

\bibitem{ChuzhoyL12}
J.~Chuzhoy and S.~Li.
\newblock A polylogarithimic approximation algorithm for edge-disjoint paths
  with congestion 2.
\newblock In {\em Proc. of IEEE FOCS}, 2012.

\bibitem{chuzhoy2017new}
Julia Chuzhoy, David~HK Kim, and Rachit Nimavat.
\newblock New hardness results for routing on disjoint paths.
\newblock In {\em Proceedings of the 49th Annual ACM SIGACT Symposium on Theory
  of Computing}, pages 86--99, 2017.

\bibitem{elkin2008lower}
Michael Elkin, Yuval Emek, Daniel~A Spielman, and Shang-Hua Teng.
\newblock Lower-stretch spanning trees.
\newblock {\em SIAM Journal on Computing}, 38(2):608--628, 2008.

\bibitem{englert2014vertex}
Matthias Englert, Anupam Gupta, Robert Krauthgamer, Harald Racke, Inbal
  Talgam-Cohen, and Kunal Talwar.
\newblock Vertex sparsifiers: New results from old techniques.
\newblock {\em SIAM Journal on Computing}, 43(4):1239--1262, 2014.

\bibitem{frank1985edge}
Andr{\'a}s Frank.
\newblock Edge-disjoint paths in planar graphs.
\newblock {\em Journal of Combinatorial Theory, Series B}, 39(2):164--178,
  1985.

\bibitem{garg2020integer}
Naveen Garg, Nikhil Kumar, and Andr{\'a}s Seb{\H{o}}.
\newblock Integer plane multiflow maximisation: Flow-cut gap and
  one-quarter-approximation.
\newblock In {\em International Conference on Integer Programming and
  Combinatorial Optimization}, pages 144--157. Springer, 2020.

\bibitem{GargVY97}
Naveen Garg, Vijay~V. Vazirani, and Mihalis Yannakakis.
\newblock Primal-dual approximation algorithms for integral flow and multicut
  in trees.
\newblock {\em Algorithmica}, 18(1):3--20, 1997.

\bibitem{gupta2004cuts}
A.~Gupta, I.~Newman, Y.~Rabinovich, and A.~Sinclair.
\newblock Cuts, trees and $\ell_1$-embeddings of graphs.
\newblock {\em Combinatorica}, 24(2):233--269, 2004.

\bibitem{gupta2001steiner}
Anupam Gupta.
\newblock Steiner points in tree metrics don't (really) help.
\newblock In {\em SODA}, volume~1, pages 220--227, 2001.

\bibitem{guruswami2003near}
V.~Guruswami, S.~Khanna, R.~Rajaraman, B.~Shepherd, and M.~Yannakakis.
\newblock {Near-optimal hardness results and approximation algorithms for
  edge-disjoint paths and related problems}.
\newblock {\em Journal of Computer and System Sciences}, 67(3):473--496, 2003.

\bibitem{harrelson2003polynomial}
Chris Harrelson, Kirsten Hildrum, and Satish Rao.
\newblock A polynomial-time tree decomposition to minimize congestion.
\newblock In {\em Proceedings of the fifteenth annual ACM symposium on Parallel
  algorithms and architectures}, pages 34--43, 2003.

\bibitem{huang2020approximation}
Chien-Chung Huang, Mathieu Mari, Claire Mathieu, Kevin Schewior, and Jens
  Vygen.
\newblock An approximation algorithm for fully planar edge-disjoint paths.
\newblock {\em arXiv preprint arXiv:2001.01715}, 2020.

\bibitem{kawarabayashi2018all}
Ken-ichi Kawarabayashi and Yusuke Kobayashi.
\newblock All-or-nothing multicommodity flow problem with bounded fractionality
  in planar graphs.
\newblock {\em SIAM Journal on Computing}, 47(4):1483--1504, 2018.

\bibitem{KleinbergT98}
J.~Kleinberg and E.~Tardos.
\newblock Approximations for the disjoint paths problem in high-diameter planar
  networks.
\newblock {\em {Journal of Computers and System Sciences}}, 57(1):61--73, 1998.

\bibitem{Okamura81}
Haruko Okamura and P.~D. Seymour.
\newblock Multicommodity flows in planar graphs.
\newblock {\em Journal of Combinatorial Theory, Series B}, 31(1):75--81,
  1981/8.

\bibitem{rabinovich1998lower}
Yuri Rabinovich and Ran Raz.
\newblock Lower bounds on the distortion of embedding finite metric spaces in
  graphs.
\newblock {\em Discrete \& Computational Geometry}, 19(1):79--94, 1998.

\bibitem{racke02}
H.~R\"acke.
\newblock Minimizing congestion in general networks.
\newblock In {\em Proc. of IEEE FOCS}, pages 43--52, 2002.

\bibitem{Racke08}
Harald R{\"a}cke.
\newblock Optimal hierarchical decompositions for congestion minimization in
  networks.
\newblock In {\em STOC}, pages 255--264, 2008.

\bibitem{racke2014improved}
Harald R{\"a}cke and Chintan Shah.
\newblock Improved guarantees for tree cut sparsifiers.
\newblock In {\em European Symposium on Algorithms}, pages 774--785. Springer,
  2014.

\bibitem{seguin2020maximum}
L.~Seguin-Charbonneau and F.B. Shepherd.
\newblock Maximum edge-disjoint paths in planar graphs with congestion 2.
\newblock {\em Math Programming}, 2020.

\bibitem{seymour1981matroids}
Paul~D Seymour.
\newblock Matroids and multicommodity flows.
\newblock {\em European Journal of Combinatorics}, 2(3):257--290, 1981.

\bibitem{tasos}
Anastasios Sidiropoulos.
\newblock Private communication, 2014.

\end{thebibliography}


\appendix

\section{Converting to Integer Weight Cut-Conservative Approximators}
\label{sec:extend}

Section~\ref{sec:non-conservative} described an algorithm to get
\[
    \frac14 d(x,y) \le d^*(x,y) \le 2 d(x,y)
    \mbox{ for leaves $x,y\in L$} .
\]
However, the assigned weights were not necessarily integers.
In this section, we modify the algorithm to assign integer weights,
and achieve the following bound:
\begin{equation}
    \frac1{14} d(x,y) \le d^*(x,y) \le d(x,y)
    \mbox{ for leaves $x,y\in L$} .
    \label{int-bounds}
\end{equation}

The algorithm will be the same as before except for two changes.
First, we choose $r_1,...,r_k$ to be the points that are exactly
distance $\lceil{h(r)/2}\rceil$ away from $r$. Second, we assign
a weight of $\max\{\lfloor{h(r_i)/2}\rfloor,1\}$ to the first edge
on the path from $r_i$ to $r$ for each $i=2,...,k$ (Figure \ref{fig:algorithm-integers}).
Note that this algorithm also terminates because the length of the
longest path from the root to a leaf decreases by at least half the
length of the shortest edge incident to a leaf in each iteration.


\begin{figure}
\centering
\begin{tikzpicture}[x=0.75pt,y=0.75pt,yscale=-1,xscale=1]

\draw  [fill={rgb, 255:red, 0; green, 0; blue, 0 }  ,fill opacity=1 ] (145.38,47.73) .. controls (145.38,45.9) and (147.12,44.42) .. (149.27,44.42) .. controls (151.41,44.42) and (153.15,45.9) .. (153.15,47.73) .. controls (153.15,49.56) and (151.41,51.04) .. (149.27,51.04) .. controls (147.12,51.04) and (145.38,49.56) .. (145.38,47.73) -- cycle ;
\draw  [dash pattern={on 0.84pt off 2.51pt}]  (5.26,148.48) -- (301.05,148.48) ;
\draw  [fill={rgb, 255:red, 0; green, 0; blue, 0 }  ,fill opacity=1 ] (56.64,148.48) .. controls (56.64,146.65) and (58.38,145.16) .. (60.53,145.16) .. controls (62.68,145.16) and (64.42,146.65) .. (64.42,148.48) .. controls (64.42,150.31) and (62.68,151.79) .. (60.53,151.79) .. controls (58.38,151.79) and (56.64,150.31) .. (56.64,148.48) -- cycle ;
\draw  [fill={rgb, 255:red, 0; green, 0; blue, 0 }  ,fill opacity=1 ] (219.87,149.74) .. controls (219.87,147.91) and (221.61,146.42) .. (223.76,146.42) .. controls (225.91,146.42) and (227.65,147.91) .. (227.65,149.74) .. controls (227.65,151.56) and (225.91,153.05) .. (223.76,153.05) .. controls (221.61,153.05) and (219.87,151.56) .. (219.87,149.74) -- cycle ;
\draw  [fill={rgb, 255:red, 0; green, 0; blue, 0 }  ,fill opacity=1 ] (158.38,148.48) .. controls (158.38,146.65) and (160.12,145.16) .. (162.27,145.16) .. controls (164.41,145.16) and (166.15,146.65) .. (166.15,148.48) .. controls (166.15,150.31) and (164.41,151.79) .. (162.27,151.79) .. controls (160.12,151.79) and (158.38,150.31) .. (158.38,148.48) -- cycle ;
\draw  [fill={rgb, 255:red, 0; green, 0; blue, 0 }  ,fill opacity=1 ] (174.96,98.1) .. controls (174.96,96.28) and (176.7,94.79) .. (178.84,94.79) .. controls (180.99,94.79) and (182.73,96.28) .. (182.73,98.1) .. controls (182.73,99.93) and (180.99,101.42) .. (178.84,101.42) .. controls (176.7,101.42) and (174.96,99.93) .. (174.96,98.1) -- cycle ;
\draw    (149.27,47.73) -- (178.84,98.1) ;
\draw    (60.53,148.48) -- (149.27,47.73) ;
\draw    (162.27,148.48) -- (180.32,95.59) ;
\draw    (225.24,147.22) -- (180.32,95.59) ;
\draw   (60.53,148.48) -- (96.68,212.09) -- (24.38,212.09) -- cycle ;
\draw   (162.27,148.48) -- (189.87,243.57) -- (134.66,243.57) -- cycle ;
\draw   (223.65,149.74) -- (251.15,212.27) -- (196.15,212.27) -- cycle ;
\draw    (301.41,45.14) -- (301.05,148.48) ;
\draw [shift={(301.05,148.48)}, rotate = 270.2] [color={rgb, 255:red, 0; green, 0; blue, 0 }  ][line width=0.75]    (0,5.59) -- (0,-5.59)   ;
\draw [shift={(301.41,45.14)}, rotate = 270.2] [color={rgb, 255:red, 0; green, 0; blue, 0 }  ][line width=0.75]    (0,5.59) -- (0,-5.59)   ;

\draw (43.21,124.3) node [anchor=north west][inner sep=0.75pt]    {$r_{1}$};
\draw (136.94,149.24) node [anchor=north west][inner sep=0.75pt]    {$r_{2}$};
\draw (234.48,147.79) node [anchor=north west][inner sep=0.75pt]    {$r_{3}$};
\draw (157.17,32.13) node [anchor=north west][inner sep=0.75pt]    {$r$};
\draw (310.76,76.38) node [anchor=north west][inner sep=0.75pt]    {$\left\lceil \frac{h( r)}{2}\right\rceil $};
\draw (120.92,107.48) node [anchor=north west][inner sep=0.75pt]  [font=\small]  {$\left\lfloor \frac{h( r_{2})}{2}\right\rfloor $};
\draw (212.64,100.48) node [anchor=north west][inner sep=0.75pt]  [font=\small]  {$\left\lfloor \frac{h( r_{3})}{2}\right\rfloor $};
\draw (169.88,64.56) node [anchor=north west][inner sep=0.75pt]    {$0$};
\draw (92.65,83.7) node [anchor=north west][inner sep=0.75pt]    {$0$};
\draw (41.25,190.2) node [anchor=north west][inner sep=0.75pt]    {$T( r_{1})$};
\draw (144.19,219.15) node [anchor=north west][inner sep=0.75pt]    {$T( r_{2})$};
\draw (205.09,191.39) node [anchor=north west][inner sep=0.75pt]    {$T( r_{3})$};
\end{tikzpicture}
\caption{
The algorithm assigns integer weights to the edges above $r_1,...,r_k$,
and is run recursively on the subtrees $T(r_1),...,T(r_k)$.
}
\label{fig:algorithm-integers}
\end{figure}
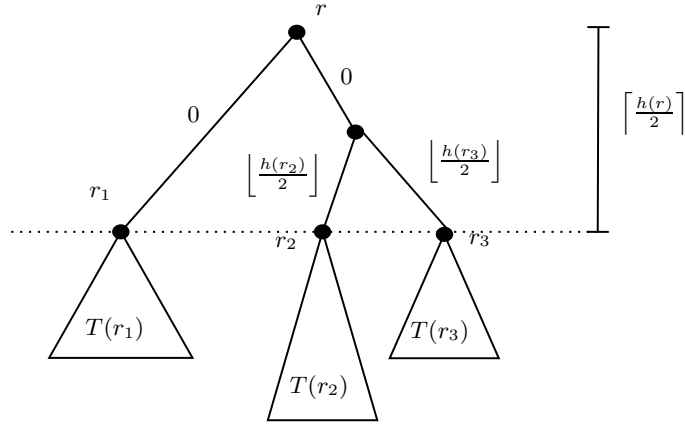

Since we assign 0 weight to edges on the path from $r_1$ to $r$,
there exists a 0 weight path from each node to a leaf.
It remains to prove that the assigned weights satisfy the bound.

To prove the upper bound $d^*(x,y)\le d(x,y)$, we need the following
lemma.

\begin{lemma} \label{int-upper-lemma}
For each leaf $x\in L$,
\begin{equation}
2d^*(x, r) \le 2d(x, r) - h(r) + 2 . \label{int-upper-bound}
\end{equation}
Furthermore, if $h(r)\le1$, we have
\begin{equation}
d^*(x, r) \le d(x, r) . \label{strong-upper-bound}
\end{equation}
\end{lemma}

\begin{proof}
We prove the upper bounds in \eqref{int-upper-bound} and \eqref{strong-upper-bound}
simultaneously by induction. We are done if $T$ is a single node, otherwise
we consider a few cases depending on whether $x\in T(r_1)$ and whether $h(r_i)\ge2$.

\textbf{Case 1} (inequality \eqref{int-upper-bound}):
$x\in T(r_1)$.
\begin{align*}
2d^*(x, r)
&= 2d^*(x, r_1) + 2d^*(r_1, r)
&& \textrm{($r_1$ is between $x$ and $r$)} \\
&= 2d^*(x, r_1)
&& \textrm{(by definition of $u^*$)} \\
&\le 2d(x, r_1) - h(r_1) + 2
&& \textrm{(by induction with \eqref{int-upper-bound})} \\
&= 2d(x, r) - 2d(r, r_1) - h(r_1) + 2
&& \textrm{($r_1$ is between $x$ and $r$)} \\
&= 2d(x, r) - (\lceil{h(r)/2}\rceil + h(r)) + 2
&& \textrm{(since $d(r_1,r)+h(r_1)=h(r)$)} \\
&\le 2d(x, r) - h(r) + 2
\end{align*}

\textbf{Case 2} (inequality \eqref{int-upper-bound}):
$x\in T(r_i)$ for some $i\neq1$, and $h(r_i)\ge2$.
Note that this means $\lfloor{h(r_i)/2}\rfloor\ge1$.
\begin{align*}
2d^*(x, r)
&= 2d^*(x, r_i) + 2d^*(r_i, r)
&& \textrm{($r_i$ is between $x$ and $r$)} \\
&= 2d^*(x, r_i) + 2\lfloor{h(r_i)/2}\rfloor
&& \textrm{(by definition of $u^*$)} \\
&\le 2d(x, r_i) - h(r_i) + 2 + h(r_i) 
&& \textrm{(by induction with \eqref{int-upper-bound})} \\
&= 2d(x, r) - 2d(r_i, r) + 2
&& \textrm{($r_i$ is between $x$ and $r$)} \\
&= 2d(x, r) - 2\lceil{h(r)/2}\rceil + 2
&& \textrm{(since $d(r_i, r) = \lceil{h(r)/2}\rceil$)} \\
&\le 2d(x, r) - h(r) + 2
\end{align*}

\textbf{Case 3} (inequality \eqref{int-upper-bound}):
$x\in T(r_i)$ and $h(r_i)\le1$ for some $i\neq1$.
\begin{align*}
2d^*(x, r)
&= 2d^*(x, r_i) + 2d^*(r_i, r)
&& \textrm{($r_i$ is between $x$ and $r$)} \\
&= 2d^*(x, r_i) + 2
&& \textrm{(by definition of $u^*$)} \\
&\le 2d(x, r_i) + 2
&& \textrm{(by induction with \eqref{strong-upper-bound})} \\
&= 2d(x, r) - 2d(r_i, r) + 2
&& \textrm{($r_i$ is between $x$ and $r$)} \\
&= 2d(x, r) - 2\lceil{h(r)/2}\rceil + 2
&& \textrm{($d(r_i, r) = \lceil{h(r)/2}\rceil$ by definition of $r_i$)} \\
&\le 2d(x, r) - h(r) + 2
\end{align*}

For the remaining cases, we assume that $h(r)=1$,
so $d(r, r_i) = 1$ for each $i$.

\textbf{Case 4} (inequality \eqref{strong-upper-bound}):
$x\in T(r_1)$. Since $h(r)=1$, we necessarily have $x=r_1$ and $d(x,r)=1$,
so
\[
d^*(x, r) = 0 \le 1 = d(x, r) .
\]

\textbf{Case 5} (inequality \eqref{strong-upper-bound}):
$x\in T(r_i)$ for some $i\neq1$ and $h(r_i)\ge2$.
\begin{align*}
2d^*(x, r)
&= 2d^*(x, r_i) + 2d^*(r_i, r)
&& \textrm{($r_i$ is between $x$ and $r$)} \\
&= 2d^*(x, r_i) + 2\lfloor{h(r_i)/2}\rfloor
&& \textrm{(by definition of $u^*$)} \\
&\le 2d(x, r_i) - h(r_i) + 2 + h(r_i) 
&& \textrm{(by induction with \eqref{int-upper-bound})} \\
&= 2d(x, r)
&& \textrm{(because $d(r,r_i)=1$)}
\end{align*}

\textbf{Case 6} (inequality \eqref{strong-upper-bound}):
$x\in T(r_i)$ and $h(r_i)\le1$ for some $i\neq1$.
\begin{align*}
2d^*(x, r)
&= 2d^*(x, r_i) + 2d^*(r_i, r)
&& \textrm{($r_i$ is between $x$ and $r$)} \\
&= 2d^*(x, r_i) + 2
&& \textrm{(by definition of $u^*$)} \\
&\le 2d(x, r_i) + 2
&& \textrm{(by induction with \eqref{strong-upper-bound})} \\
&= 2d(x, r)
&& \textrm{(because $d(r,r_i)=1$)}
\end{align*}

\end{proof}

To prove the lower bound $d(x,y)\le 14d^*(x,y)$, we need the following
lemma.

\begin{lemma} \label{int-lower-lemma}
For each leaf $x$,
\begin{equation}
    d(x, r) \le 3d^*(x, r) + h(r) . \label{int-lower-bound}
\end{equation}
\end{lemma}

\begin{proof}
We prove the lower bound in \eqref{int-lower-bound} by induction.

\textbf{Case 1}: $x \in T(r_1)$.
\begin{align*}
3d^*(x, r)
&= 3d^*(x, r_1) + 3d^*(r_1, r)
&& \textrm{($r_1$ is between $x$ and $r$)} \\
&= 3d^*(x, r_1)
&& \textrm{(by definition of $u^*$)} \\
&\ge d(x, r_1) - h(r_1)
&& \textrm{(by induction)} \\
&= d(x, r) - d(r, r_1) - h(r_1)
&& \textrm{($r_1$ is between $x$ and $r$)} \\
&= d(x, r) - h(r)
&& \textrm{(by definition of $r_1$)}
\end{align*}
Rearranging gives the desired inequality.

\textbf{Case 2}: $x \in T(r_i)$ for some $i\neq1$.
\begin{align*}
3d^*(x, r)
&= 3d^*(x, r_i) + 3d^*(r_i, r)
&& \textrm{($r_i$ is between $x$ and $r$)} \\
&\ge d(x, r_i) - h(r_i) + 3d^*(r_i, r)
&& \textrm{(by induction)} \\
&= d(x, r_i) - h(r_i) + 3\max\{1, \lfloor{h(r_i)/2}\rfloor\}
&& \textrm{(by definition of $u^*$)} \\
&= d(x, r_i) - h(r_i) + 1 + 2\lfloor{h(r_i)/2}\rfloor \\
&\ge d(x, r_i) \\
&= d(x, r) - d(r_i, r)
&& \textrm{($r_i$ is between $x$ and $r$)} \\
&\ge d(x, r) - h(r)
&& \textrm{(since $d(r_i,r)\le h(r)$)}
\end{align*}
Rearranging gives the desired inequality.

\end{proof}

Finally, we prove the bounds in \eqref{int-bounds} by induction.
Let $x,y\in L$ be two leaves of $T$. Suppose that
$x\in T(r_i)$ and $y\in T(r_j)$. By induction, we may assume that
$i\neq j$ (otherwise we may consider the subtree rooted at $r_i$
since the algorithm is recursive). Without loss of generality,
suppose that $i<j$.

We prove the upper bound. Let $c$ be the lowest common ancestor
of $x$ and $y$. Since $r_i\neq r_j$, the node $c$ is also an
ancestor of $r_i$ and $r_j$, so $d(r_i,c)\ge1$ and $d(r_j,c)\ge1$.
First we show that $d^*(x,c)\le d(x,c)$ by considering two cases.

\textbf{Case 1}: $h(r_i)\ge2$.
\begin{align*}
2d^*(x,c)
&= 2d^*(x,r_i) + 2d^*(r_i,c)
&& \textrm{($r_i$ is between $x$ and $c$)} \\
&\le 2d(x, r_i) - h(r_i) + 2 + 2d^*(r_i, c)
&& \textrm{(by upper bound \eqref{int-upper-bound})} \\
&\le 2d(x, r_i) - h(r_i) + 2 + 2\lfloor{h(r_i)/2}\rfloor
&& \textrm{(by definition of $u^*$)} \\
&\le 2d(x, r_i) + 2
&& \textrm{(since $2\lfloor{h(r_i)/2}\rfloor\le h(r_i)$)} \\
&\le 2d(x, r_i) + 2d(r_i, c)
&& \textrm{(since $d(r_i,c)\ge1$)} \\
&= 2d(x, c)
\end{align*}

\textbf{Case 2}: $h(r_i)\le1$.
\begin{align*}
2d^*(x,c)
&= 2d^*(x,r_i) + 2d^*(r_i,c)
&& \textrm{($r_i$ is between $x$ and $c$)} \\
&\le 2d(x, r_i) + 2d^*(r_i, c)
&& \textrm{(by upper bound \eqref{strong-upper-bound})} \\
&= 2d(x, r_i) + 2
&& \textrm{(by definition of $u^*$)} \\
&\le 2d(x, c)
&& \textrm{(since $d(r_i,c)\ge1$)}
\end{align*}

The proof that $d^*(y,c)\le2d(y,c)$ is the same. We conclude that
\[
d^*(x,y) = d^*(x,c) + d^*(y,c)
\le d(x,c) + d(y,c) = d(x,y) ,
\]
and the upper bound is proved.


We prove the lower bound. First note that $i<j$ means
\[
14d^*(r_i, r_j)
\ge 14\max\{1, \lfloor{h(r_j)/2}\rfloor)\}
\ge 6\cdot 2 + 8\cdot \lfloor{h(r_j)/2}\rfloor
\ge 2 + 4h(r_j) .
\]
Then, we can compute
\begin{align*}
d(x, y)
&= d(x, r_i) + d(r_i, r_j) + d(r_j, y) \\
&\le d(x, r_i) + d(y, r_j) + d(r, r_i) + d(r, r_j) \\
&\le d(x, r_i) + d(y, r_j) + h(r_i) + h(r_j) + 2
&&\textrm{(since $d(r,r_i)\le h(r_i)+1$)} \\
&\le 3d^*(x, r_i) + 3d^*(y, r_j) + 2h(r_i) + 2h(r_j) + 2
&& \textrm{(by \eqref{int-lower-bound})} \\
&\le 14d^*(x, r_i) + 14d^*(y, r_j) + 4h(r_j) + 2 \\
&\le 14d^*(x, y) - 14d^*(r_i, r_j) + 4h(r_j) + 2 \\
&\le 14d^*(x, y) - 2 - 4h(r_j) + 4h(r_j) + 2 \\
&\le 14d^*(x, y) .
\end{align*}

This proves inequality \eqref{int-bounds}.

 \section{Lower bound on  Congestion in the Exact Weight Model} 
\label{sec:lb}

In this section we prove Theorem~\ref{thm:lowerbound}.

\begin{proof}
Let $G$ be an undirected unit capacity graph with $3\times2^n$ nodes defined as follows.
Label the nodes of $G$ with the integers $0$ to $3\times2^n-1$,
and arrange the nodes in a circle so that node $i+1$ comes after node $i$.
First add the three edges $(0,2^n)$, $(2^n,2^{n+1})$, and $(2^{n+1},0)$.
Then, recursively for each edge $(u,v)$ that is not between adjacent nodes on the circle,
add the two edges $(u,(u+v)/2)$ and $((u+v)/2,v)$ (Figure \ref{fig:lb-congestion}).
We show that any spanning tree of $G$ has $\Omega(n)$ congestion.

\begin{figure}
\centering
\begin{tikzpicture}[x=0.75pt,y=0.75pt,yscale=-1,xscale=1]

\draw   (36,148.18) .. controls (36,86.46) and (86.03,36.43) .. (147.75,36.43) .. controls (209.47,36.43) and (259.5,86.46) .. (259.5,148.18) .. controls (259.5,209.9) and (209.47,259.93) .. (147.75,259.93) .. controls (86.03,259.93) and (36,209.9) .. (36,148.18) -- cycle ;
\draw    (147.75,36.43) -- (236.75,217.18) ;
\draw  [fill={rgb, 255:red, 0; green, 0; blue, 0 }  ,fill opacity=1 ] (144,36.43) .. controls (144,34.36) and (145.68,32.68) .. (147.75,32.68) .. controls (149.82,32.68) and (151.5,34.36) .. (151.5,36.43) .. controls (151.5,38.5) and (149.82,40.18) .. (147.75,40.18) .. controls (145.68,40.18) and (144,38.5) .. (144,36.43) -- cycle ;
\draw  [fill={rgb, 255:red, 0; green, 0; blue, 0 }  ,fill opacity=1 ] (49,91.18) .. controls (49,89.11) and (50.68,87.43) .. (52.75,87.43) .. controls (54.82,87.43) and (56.5,89.11) .. (56.5,91.18) .. controls (56.5,93.25) and (54.82,94.93) .. (52.75,94.93) .. controls (50.68,94.93) and (49,93.25) .. (49,91.18) -- cycle ;
\draw  [fill={rgb, 255:red, 0; green, 0; blue, 0 }  ,fill opacity=1 ] (144,259.93) .. controls (144,257.86) and (145.68,256.18) .. (147.75,256.18) .. controls (149.82,256.18) and (151.5,257.86) .. (151.5,259.93) .. controls (151.5,262) and (149.82,263.68) .. (147.75,263.68) .. controls (145.68,263.68) and (144,262) .. (144,259.93) -- cycle ;
\draw  [fill={rgb, 255:red, 0; green, 0; blue, 0 }  ,fill opacity=1 ] (51,210.18) .. controls (51,208.11) and (52.68,206.43) .. (54.75,206.43) .. controls (56.82,206.43) and (58.5,208.11) .. (58.5,210.18) .. controls (58.5,212.25) and (56.82,213.93) .. (54.75,213.93) .. controls (52.68,213.93) and (51,212.25) .. (51,210.18) -- cycle ;
\draw  [fill={rgb, 255:red, 0; green, 0; blue, 0 }  ,fill opacity=1 ] (233,217.18) .. controls (233,215.11) and (234.68,213.43) .. (236.75,213.43) .. controls (238.82,213.43) and (240.5,215.11) .. (240.5,217.18) .. controls (240.5,219.25) and (238.82,220.93) .. (236.75,220.93) .. controls (234.68,220.93) and (233,219.25) .. (233,217.18) -- cycle ;
\draw  [fill={rgb, 255:red, 0; green, 0; blue, 0 }  ,fill opacity=1 ] (238,92.18) .. controls (238,90.11) and (239.68,88.43) .. (241.75,88.43) .. controls (243.82,88.43) and (245.5,90.11) .. (245.5,92.18) .. controls (245.5,94.25) and (243.82,95.93) .. (241.75,95.93) .. controls (239.68,95.93) and (238,94.25) .. (238,92.18) -- cycle ;
\draw  [fill={rgb, 255:red, 0; green, 0; blue, 0 }  ,fill opacity=1 ] (92,50.18) .. controls (92,48.11) and (93.68,46.43) .. (95.75,46.43) .. controls (97.82,46.43) and (99.5,48.11) .. (99.5,50.18) .. controls (99.5,52.25) and (97.82,53.93) .. (95.75,53.93) .. controls (93.68,53.93) and (92,52.25) .. (92,50.18) -- cycle ;
\draw  [fill={rgb, 255:red, 0; green, 0; blue, 0 }  ,fill opacity=1 ] (32.25,151.93) .. controls (32.25,149.86) and (33.93,148.18) .. (36,148.18) .. controls (38.07,148.18) and (39.75,149.86) .. (39.75,151.93) .. controls (39.75,154) and (38.07,155.68) .. (36,155.68) .. controls (33.93,155.68) and (32.25,154) .. (32.25,151.93) -- cycle ;
\draw  [fill={rgb, 255:red, 0; green, 0; blue, 0 }  ,fill opacity=1 ] (90,246.18) .. controls (90,244.11) and (91.68,242.43) .. (93.75,242.43) .. controls (95.82,242.43) and (97.5,244.11) .. (97.5,246.18) .. controls (97.5,248.25) and (95.82,249.93) .. (93.75,249.93) .. controls (91.68,249.93) and (90,248.25) .. (90,246.18) -- cycle ;
\draw  [fill={rgb, 255:red, 0; green, 0; blue, 0 }  ,fill opacity=1 ] (202,54.18) .. controls (202,52.11) and (203.68,50.43) .. (205.75,50.43) .. controls (207.82,50.43) and (209.5,52.11) .. (209.5,54.18) .. controls (209.5,56.25) and (207.82,57.93) .. (205.75,57.93) .. controls (203.68,57.93) and (202,56.25) .. (202,54.18) -- cycle ;
\draw  [fill={rgb, 255:red, 0; green, 0; blue, 0 }  ,fill opacity=1 ] (195,249.18) .. controls (195,247.11) and (196.68,245.43) .. (198.75,245.43) .. controls (200.82,245.43) and (202.5,247.11) .. (202.5,249.18) .. controls (202.5,251.25) and (200.82,252.93) .. (198.75,252.93) .. controls (196.68,252.93) and (195,251.25) .. (195,249.18) -- cycle ;
\draw  [fill={rgb, 255:red, 0; green, 0; blue, 0 }  ,fill opacity=1 ] (255.75,151.93) .. controls (255.75,149.86) and (257.43,148.18) .. (259.5,148.18) .. controls (261.57,148.18) and (263.25,149.86) .. (263.25,151.93) .. controls (263.25,154) and (261.57,155.68) .. (259.5,155.68) .. controls (257.43,155.68) and (255.75,154) .. (255.75,151.93) -- cycle ;
\draw    (236.75,217.18) -- (147.75,259.93) ;
\draw    (241.75,92.18) -- (236.75,217.18) ;
\draw    (241.75,92.18) -- (147.75,36.43) ;
\draw    (147.75,259.93) -- (54.75,210.18) ;
\draw    (54.75,210.18) -- (52.75,91.18) ;
\draw    (147.75,36.43) -- (52.75,91.18) ;
\draw    (54.75,210.18) -- (236.75,217.18) ;
\draw    (54.75,210.18) -- (147.75,36.43) ;

\draw (142,10.02) node [anchor=north west][inner sep=0.75pt]    {$0$};
\draw (248,214.02) node [anchor=north west][inner sep=0.75pt]    {$4$};
\draw (29,209.02) node [anchor=north west][inner sep=0.75pt]    {$8$};
\end{tikzpicture}
\caption{
The congestion lower bound graph for $n=2$.
}
\label{fig:lb-congestion}
\end{figure}
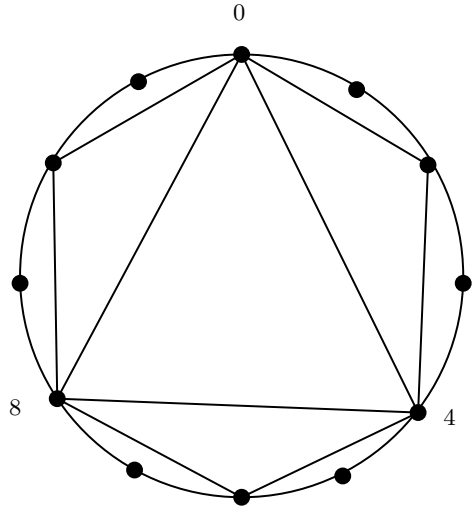

Since $G$ is planar and 2-connected, by Andersen \& Feige \cite{andersen2009interchanging},
we know that every spanning tree of $G$ with congestion $\rho$ gives a spanning tree of
the planar dual of $G$ with stretch at most $\rho+1$. Hence,
it suffices to show that every spanning tree of the planar dual of $G$ has $\Omega(n)$ stretch.

Let $v\in V(\bar G)$ be the node that corresponds to the outer face of $G$.
Let $u_k\in V(\bar G)$ be the node that corresponds to the face of $G$ bounded by
the triangle $0, 2^k, 2^{k+1}$ for $k=0,...,n$.
We show that the distance between $v$ and $u_n$ is $n+1$.
By symmetry, any path from $v$ to $u_n$ has the same length,
so we just need to find the length of one such path.
Note that each $u_k$ is adjacent to $u_{k-1}$ for $k=1,...,n$, and $u_0$ is adjacent to $v$,
so $v, u_0, u_1, ..., u_n$ is a path from $v$ to $u_n$ with length $n+1$.
Also note that there are at least two node disjoint paths from $v$ to $u_n$,
so the shortest cycle containing both $v$ and $u_n$ has length $2n+2$.
Let $C$ denote this cycle. By  \cite{rabinovich1998lower},
any dominating tree of $C$ has stretch $\Omega(|V(C)|)=\Omega(n)$.
In particular, any spanning tree of $\bar G$ is a dominating tree of $C$,
so any spanning tree of $\bar G$ has stretch $\Omega(n)$.

We conclude that any spanning tree of $G$ has congestion $\Omega(n)=\Omega(\log|V(G)|)$.
\end{proof}

\end{document}